\documentclass[11pt]{article}

\usepackage[hidelinks]{hyperref}

\usepackage[numbers]{natbib}
\usepackage[algo2e,ruled]{algorithm2e}
\usepackage{mathrsfs}
\usepackage{amssymb}
\usepackage{amsmath}
\usepackage{latexsym}
\usepackage{amsfonts}
\usepackage{graphicx}
\usepackage{amsthm}
\usepackage{subcaption}

\makeatletter
\newtheorem*{rep@theorem}{\rep@title}
\newcommand{\newreptheorem}[2]{%
    \newenvironment{rep#1}[1]{%
        \def\rep@title{#2 \ref{##1}}%
        \begin{rep@theorem}}%
        {\end{rep@theorem}}}
\makeatother

\newcommand{\ignore}[1]{}

\newtheorem{theorem}{Theorem}
\newtheorem{lemma}{Lemma}
\newtheorem{definition}{Definition}
\newtheorem{corollary}{Corollary}
\newtheorem{claim}{Claim}

\newreptheorem{theorem}{Theorem}
\newreptheorem{lemma}{Lemma}

\begin{document}
    \title{Coalitional Permutation Manipulations in the Gale-Shapley Algorithm}
    \author{
        Weiran Shen\thanks{To whom correspondence should be addressed.}\\
        Gaoling School of Artificial Intelligence\\
        Renmin University of China\\
        \texttt{shenweiran@ruc.edu.cn}
        \and
        Yuan Deng\\
        Google Research\\
        \texttt{yuandeng0124@gmail.com}
        \and
        Pingzhong Tang\\
        Institute for Interdisciplinary Information Sciences\\
        Tsinghua University\\
        \texttt{kenshinping@gmail.com}
    }
    \date{}
    \maketitle

\begin{abstract}
    In this paper, we consider permutation manipulations by any subset of women in the men-proposing version of the Gale-Shapley algorithm. This paper is motivated by the college admissions process in China. Our results also answer an open problem on what can be achieved by permutation manipulations. We present an efficient algorithm to find a strategy profile such that the induced matching is stable and Pareto-optimal (in the set of all achievable stable matchings) while the strategy profile itself is inconspicuous. Surprisingly, we show that such a strategy profile actually forms a Nash equilibrium of the manipulation game. In the end, we show that it is NP-complete to find a manipulation that is strictly better for all members of the coalition. This result demonstrates a sharp contrast between weakly better off outcomes and strictly better-off outcomes.
\end{abstract}

\section{Introduction}
\label{sec:intro}

Stability has been a central concept in economic design, ever since the seminal work by \citet{gale1962college}. Intensive research has been done over the years. A variety of applications of this problem have also been developed, ranging from college admissions and school choice \citep{abdulkadiroglu2003school, abdulkadirouglu2005new, gale1962college} to centralized kidney exchange programs \citep{abraham2007clearing, roth2004kidney, roth2005kidney,Liu14} to hospitals-residents matchings \citep{irving2009finding, irving2000hospitals, roth1996national} to recently proposed water right trading \cite{Liu16,Zhan17}.


In the standard stable matching model, there is a set of men and a set of women. Each agent has a preference list over a subset of the opposite sex. A matching between men and women is stable if no two agents prefer to match with each other than their designated partners. \citet{gale1962college} put forward an algorithm, aka the Gale-Shapley algorithm, that computes a stable matching in $O(n^2)$ time. The algorithm (men-proposing version) proceeds in multiple rounds. At each round, each man proposes to his favorite woman that has not rejected him yet; and each woman keeps her favorite proposal, if any, and rejects all others. The algorithm iterates until no further proposal can be made. 

The algorithm enjoys many desirable properties. It is well-known that the matching returned by the algorithm is preferred by every man to any other stable matching, hence called the M-optimal (for men-optimal) matching. It is also known that all stable matchings form a lattice defined by such a preference relation and the M-optimal matching is the greatest element in the lattice \citep{MarriageStables}. Furthermore, men and women have strictly opposite preferences over two stable matchings: every man prefers stable matching $\mu_1$ to stable matching $\mu_2$ if and only if every woman prefers $\mu_2$ to $\mu_1$. As a result, the M-optimal matching is the W-pessimal (for women-pessimal) matching \citep{mcvitie1971stable}. The smallest element in the lattice, the W-optimal (M-pessimal) matching, can be obtained by swapping the roles of men and women.

\subsection{Motivations}
This work is motivated by the college admissions process in China~\citep{chen2017chinese}, where the stable matching model is adopted. The admissions process consists of two phases: the examination phase and the application phase. 
In the examination phase, all students are required to take the National College Entrance Examination (NCEE, aka the National Higher Education Entrance Examination), which is held nation-wide annually. 
Millions of students take the NCEE every year, and the number peaked at 10.5 millions 2008.
The NCEE contains a series of exams on different subjects. After the examination, each student receives a total score which is the sum of the scores of the subjects. The total score uniquely determines an ordering of all students, which is also the preference ordering adopted by all colleges and universities, i.e., all colleges and universities share a master preference ordering.
In the application phase, each student submits an ordered list of about 4 to 6 
intended colleges or universities. 
In the end, the Ministry of Education 
settles the applications using the student-proposing version of the Gale-Shapley algorithm.

However, a major concern of the Gale-Shapley algorithm is its non-truthfulness. While it is known that the algorithm is group strategy-proof\footnote{Precisely, {\em group strategy-proof} means no coalition manipulation can make all men in the coalition strictly better off, in this context. If considering the case where no man is worse off and at least one man is strictly better off, the Gale-Shapley algorithm is not group strategy-proof~\cite{huang2006cheating}.} for all men \citep{dubins1981machiavelli}, it is not truthful for women. In fact, \citet{roth1982economics} shows that there is no stable matching algorithm that is strategy-proof for all agents. 

Such an undesirable property gives rise to the so-called ``manipulation'' problem for the women. In China's college admissions process, besides the NCEE, some top universities are also allowed to conduct the so-called {\em independent recruitment} exams. 
The independent recruitment programs date back to 2003 and is part of the college admissions reform of China. 
The exams are extremely competitive since only a very small fraction of applicants are selected to take such exams. 
These universities promise to the students who perform well in these exams that, when applying to these universities, a certain amount of extra scores 
will be added to the their NCEE total scores. In other words, these universities can change their ordering of students by moving some students to a higher rank.

It is worth mentioning that such independent recruitment programs are initially designed so that the universities can test the students in their own ways and thus increase the admission quality. However, such programs may also give the universities the ability to change the admissions result by changing the ordering of the students in their preference ordering.

Starting from 2010, several leagues of such universities emerged, with the two most influential leagues represented by China's two major universities, the Tsinghua University (the Tsinghua league) and the Peking University (the Peking league). Each league contains universities of similar types and tiers. Thus universities of the same league attract about the same set of students, 
and they conduct the independent recruitment programs together~\citep{pku_news_2011}.
The benefits of such leagues are obvious: (1) the costs of organizing such exams are greatly reduced since they are shared by the universities; (2) the students only need to participate in one such exam instead of many. However, such leagues are widely conjectured to be beneficial to universities inside the leagues when it comes to the quality of finally admitted students, since they can jointly manipulate the admissions result to benefit them all. Besides cooperations, the universities in the same league are also faced with the problem of competition because they share a similar set of candidate students. In 2012, two top universities (Fudan University and Nankai University) quit the Peking league, both claiming that they were not able to recruit their desired students. Such leagues were urged to dissolve in 2015 by the Ministry of Education for the belief that it is unfair for universities that are not in any of the leagues.

\subsection{Results}

We study the problem where a coalition of women (universities) can manipulate the Gale-Shapley algorithm. Most existing works consider the general case where women can report any preference list (potentially incomplete) without ties. In contrast, we focus on the setting where all women must report a complete list, which indicates that women can only permute their true preferences. This type of manipulation comes directly from the independent recruitment programs in China, where the universities can only permute the ordering of the students by adding scores to some of them, but are not allowed to remove any student from the lists.
	
	We model the coalition manipulation problem as a game among the members of the coalition (called the manipulation game hereafter). We first show that a coalition of women could get worse off if they perform their optimal single-agent manipulation separately (see Table~\ref{hardexample} for details). This result confirms that there are conflicts between different universities in the same league so that they need to find a way to manipulate jointly to achieve a better outcome. 
	
	We present an efficient algorithm to find a strategy profile such that (1) the induced matching is stable with respect to the true preference, (2) the induced matching is Pareto-optimal among all stable matchings that can be achieved by coalitional permutation manipulations (We say such a matching is $S_L$-Pareto-optimal. See Definition \ref{def:pareto}), and (3) the strategy profile is inconspicuous, where inconspicuous manipulations are those in which each woman of the coalition only moves one man to a higher rank (Algorithm~\ref{algorithm} and Algorithm~\ref{algorithm2}). Surprisingly, we show that such a strategy profile actually forms a Nash equilibrium of the manipulation game (Theorem~\ref{thm:charOfNE}). Therefore, the strategy profile found by our algorithm captures both the cooperation and the competition among the universities in the same league. This result implies that it is computationally easy to find a ``profitable'' manipulation that is weakly better off and $S_L$-Pareto-optimal for all members of the coalition, supporting the wide conjecture that such leagues of universities can benefit from forming coalitions. 
	
	All these results confirm the belief of the Ministry of Education that such leagues of universities are unfair for other universities. In the end, we show that it is NP-complete to find a manipulation that is strictly better off for all members of the coalition (Theorem~\ref{theo:hardness_better}). This result demonstrates a sharp contrast between weakly better off outcomes and strictly better off outcomes: if a manipulation is costly so that every manipulator must be strictly better off to ensure nonnegative payoff, a coalition manipulation is unlikely to happen due to computational burdens. 
	
	Our results also give answers to the open problem raised by \citet{Gusfield:1989} on what can be produced by permutation manipulations (see also \citep{kobayashi2010cheating} and \citep{sukegawa2012preference} for more of the problem).  

\subsection{Additional Related Works}


There is a large body of literature that focuses on finding manipulations for women when fixing men's preferences in the Gale-Shapley algorithm. \citet{gale1985ms} show that it is possible for all women to strategically truncate their preference lists so that each of them is matched with their partner in the W-optimal matching, and \citet{teo2001gale} provide a polynomial time algorithm to find the optimal single-agent truncation manipulation. 


\citet{teo2001gale} study permutation manipulations, where a woman can report any permutation of her true preference list.
Their work is motivated by the primary student assignment process in Singapore. They give an efficient algorithm to compute the best manipulation for a single manipulator. Recently, \citet{ijcai2017-62} follows an early version of this paper and shows that the resulting matching from optimal singleton permutation manipulation is stable with respect to true preference lists and there exists an inconspicuous singleton manipulation which is optimal. 
\citet{gupta2015stable} extends the algorithm from \citet{teo2001gale} to the so-called $P$-stable (stable with respect to preferences $P$) Nash equilibrium setting. \citet{aziz2015susceptibility} also study permutation manipulations in a many-to-one setting, but focus on a single manipulator with quota more than one. \citet{pini2009manipulation} create a stable matching mechanism and show that for a single agent, it is computationally hard to manipulate the matching result. All the results, except for the last, do not apply to cases where a coalition of women jointly manipulate. 

There is also a line of works that consider the controlled school choice problem with diversity constraints, where the students fall into different categories and each school has minimum quotas for different quotas \citep{abdulkadiroglu2003school,echenique2015control,gonczarowski2019matching,baswana2019centralized}. \citet{ehlers2014school} show that a solution may not exist if the diversity constraints are hard. Therefore, recent works on this topic mostly focus on the setting where the diversity constraints are soft constraints \citep{hafalir2013effective,kurata2017controlled}.

\section{Preliminaries}
\label{sec:pre}   

We consider a stable matching model with a set of men $M$ and a set of women $W$, where only complete and strict preference lists are allowed.\footnote{We consider the case where men also report complete preference lists for simplicity. Our result can be generalized to the case where men may report incomplete preference lists.}  
The preference list of a man $m$, denoted by $P(m)$, in a preference profile $P$, is a strict total order $\succ^P_m$ over the set of women $W$. Let $w_1 \succ^P_m w_2$ denote that $m$ prefers $w_1$ to $w_2$ in profile $P$. Similarly, the preference list $P(w)$ of a woman $w$ is a strict total order over $M$. For simplicity, we sometimes use $\succ_w$ to denote the true preference list when it is clear from the context. 

A matching is a function $\mu : M \cup W \mapsto M \cup W$. We write $\mu(m)=w$ if a man $m$ is matched to a woman $w$. Similarly, $\mu(w)=m$ if $w$ is matched to $m$. Also, $\mu(m)=w$ if and only if $\mu(w)=m$. We will also write $\mu(m)=m$ (or $\mu(w)=w$) if $m$ (or $w$) is left unmatched. For two matchings $\mu_1$ and $\mu_2$, if for all $w \in W$, $\mu_1(w) \succeq_w \mu_2(w)$, we say $\mu_1 \succeq_W \mu_2$. 
If in a matching $\mu$, a man $m$ and a woman $w$ are not matched together, yet prefer each other to their partners in $\mu$, then $(m, w)$ is called a \emph{blocking pair}. A matching is \emph{stable} if and only if it contains no blocking pair. 

The Gale-Shapley algorithm is not truthful for women~\cite{dubins1981machiavelli}. Given a set of women manipulators, the algorithm can be thought of as a game (henceforth, the {\em manipulation game}), between them. 

\begin{definition}[Manipulation game]
Given a true preference profile $P$,  a manipulation game is a tuple $(L,\mathcal A)$, where:
\begin{enumerate}
	\item $L\subseteq W$ is the set of manipulators;
	\item $\mathcal A = \prod_{i\in L} A_i$ is the set of all possible reported preference profiles. 
\end{enumerate} 

\end{definition}

The outcome of the manipulation game (also called induced matching in this paper) is the matching resulted from the Gale-Shapley algorithm with respect to the reported preference profiles.  A manipulator's preference in this game is her true preference in $P$.

Motivated by the NCEE in China, we focus on the setting where all women must report a complete list of men, which indicates that women can only permute their true preferences in the manipulation.

\begin{definition}[Permutation manipulation]\label{permutation}
Let $\mathbb{O}$ be the set of strict total orders over $M$. In permutation manipulations, $A_i=\mathbb{O}$, $\forall i \in L$.
\end{definition}



%


Let $P(M)=(P(m) : m \in M)$ be the true preference profile of all men. Similarly, denote the true preference profiles for all women, all manipulators and all non-manipulators by $P(W)$, $P(L)$ and $P(N)$, respectively, where $N = W \setminus L$ is the set of non-manipulators. Thus the overall true preference profile is $P = (P(M), P(N), P(L))$. Denote by $S(P(M), P(W))$ the set of all stable matchings under profile $(P(M), P(W))$. Let $S_L(P(M), P(W)) \subseteq S(P(M), P(W))$ be the set of all stable matchings that can be achieved by a coalition manipulation of $L$. Note that all matchings in $S_L(P(M), P(W))$ are stable with respect to the true preference profile, since $S_L(P(M), P(W))$ is a subset of $S(P(M), P(W))$ We sometimes write $S_L$ for short when $(P(M), P(W))$ is clear from the context. We define Pareto-optimality within the set $S_L$.

\begin{definition}[$S_L$-Pareto-optimal matching]
    \label{def:pareto}
    A matching $\mu$ is $S_L$-Pareto-optimal if $\mu\in S_L$ and there is no $\mu'\in S_L$ such that all manipulators are weakly better off and at least one is strictly better off.  
\end{definition}

    We say a strategy profile $P(L)$ of a manipulation game is $S_L$-Pareto-optimal if its induced matching is $S_L$-Pareto-optimal. In a manipulation game, the solution concept we are interested in is Nash equilibrium.

\begin{definition}[Nash equilibrium]
	A preference profile $P(L)= \bigcup_{l\in L} P(l)$ of a manipulation game is a Nash equilibrium, if $~\forall l\in L$, $l$ cannot get a strictly better partner with respect to the true preference list by reporting any other preference list.
\end{definition}

Our algorithm is enabled by two special structures, the {\em rotation} \citep{gusfield1987three} and the {\em modified suitor graph}.

\subsection{Rotation}\label{subsec:rotation}

The concept of rotation was first introduced by \citet{irving1985efficient} when solving the stable roommate problem, which is a natural generalization of the stable marriage problem.

In the Gale-Shapley algorithm, if a woman $w_i$ rejects a man $m_j$, then $w_i$ must have a better partner than $m_j$ in the W-pessimal matching. Thus in any stable matching, $w_i$ cannot be matched with any man ranked below $m_j$ in $w_i$'s list. As a result, we can safely remove all impossible partners from each man or woman's preference list after each iteration of the algorithm. We call each man or woman's preference list after the removal a {\em reduced list}, and the set of all reduced lists a {\em reduced table}.

\begin{definition}[Rotation \citep{irving1985efficient}]
	A rotation $R=(m_1, m_2, \dots, m_r)$ is a sequence of men, where the first woman in the reduced list of $m_{i+1}$ is the second in that of $m_i$ ($i+1$ is taken modulo r).
\end{definition}

Note that rotations are known as {\em improvement cycles} in some literature and is useful in converting the M-optimal matching to the W-optimal matching~\citep{Ashlagi:2013:URM:2492002.2482590, gonczarowski2013manipulation, immorlica2005marriage}.

We also use $R=(\mathcal{M},\mathcal{W},\mathcal{W}')$ to represent a rotation, where $\mathcal{M}$ is the sequence of men and $\mathcal{W}$ and $\mathcal{W}'$ are the sequences of the first and the second women in $\mathcal{M}$'s reduced lists. Since $\mathcal{W}_{i+1}=\mathcal{W}'_i$ according to the definition of rotations, we write $\mathcal{W}^r=\mathcal{W}'$, where $\mathcal{W}^r$ is the sequence $\mathcal{W}$ with each woman shifted left by one position. We may sometimes use $m_i$ and $w_i$ to mean the $i$-th agent in $\mathcal{M}$ and $\mathcal{W}$ when the order is important. 

After the termination of the Gale-Shapley algorithm, one can still change the matching by eliminating rotations. The elimination of a rotation $R$ is to force each woman $w_i$ in $\mathcal{W}$ to reject her current proposer $m_i$ and let $m_i$ propose to $w_{i+1}$. It is clear that after the elimination, each woman still holds a proposal, i.e., there is still a matching between men and women. More importantly, it can be shown that the matching is stable with respect to the true preference. We say a rotation $R=(\mathcal{M},\mathcal{W},\mathcal{W}^r)$ moves $m_i$ from $w_i$ to $w_{i+1}$ and moves $w_i$ from $m_i$ to $m_{i-1}$ since after eliminating the rotation, the corresponding matching matches $m_i$ and $w_{i+1}$. 
It is known that each stable matching corresponds to a set of rotations, and there exists an order of elimination that produces the matching, which we do not discuss in detail here, but refer readers to \cite{Gusfield:1989}.

\subsection{Modified Suitor Graph}\label{subsec:suitor}

Suitor graph is another important structure in our analysis. It is initially proposed by \citet{kobayashi2009successful,kobayashi2010cheating} when considering the following problem: given a matching $\mu$ and a preference profile for all men $P(M)$, is there a profile $P(W)$ for the women, such that the M-optimal matching of the combined preference profile is $\mu$?

We change the definition to suit our setting and propose the \emph{modified suitor graph} as follows:

\begin{definition}[Modified Suitor Graph] \label{SuitorGraph}
Given a matching $\mu$, a preference profile for all men $P(M)$ and a preference profile for all non-manipulators $P(N)$, the modified suitor graph $G(P(M), P(N), \mu)$ is a directed graph $(V, E)$, which can be constructed using the following steps:
\begin{enumerate}
	\item $V=M \cup W \cup \{s\}$, where $s$ is a virtual vertex;
	\item $\forall w \in W$, add edges $(w, \mu(w))$ and $(\mu(w), w)$;
	\item Let $\delta(w) = \{m ~|~ w \succ_m \mu(m)\}$. $\forall w \in L$ and for each $m$ in $\delta(w)$, add edges $(m, w)$;
	\item $\forall w \in N$, if $\delta(w)$ is nonempty, add the edge $(m, w)$, where $m$ is $w$'s favorite in $\delta(w)$;
	\item $\forall w \in W$, if $\delta(w) = \emptyset$, add an edge $(s, w)$ to the graph;
\end{enumerate}
\end{definition}

Our definition of suitor graph is slightly different from the original ones. A detailed discussion on the differences is provided in \ref{app:suitor_graph}.

\citet{kobayashi2009successful,kobayashi2010cheating} give a characterization of the existence of a profile $P(W)$ and an $O(n^2)$ time algorithm that can be derived directly from their constructive proof. Here, we provide a similar result tailored for our setting and will be useful for later analyses.

\begin{lemma}\label{CharOfProfile}
    Given a matching $\mu$, a preference profile with $P(M)$ for all men and $P(N)$ for all non-manipulators, there exists a profile for the manipulators $P(L)$ such that $\mu$ is the M-optimal stable matching for the total preference profile $(P(M), P(N), P(L))$, if and only if for every vertex $v$ in the corresponding suitor graph $G(P(M), P(N), \mu)$, there exists a directed path from $s$ to $v$ ($s$ is the virtual vertex in the graph). Moreover, if such a $P(L)$ exists, it can be constructed in $O(n^2)$.
\end{lemma}
The proof of the above lemma is similar to the arguments by \citet{kobayashi2009successful}, and thus is omitted in this paper.

%

\section{$S_L$-Pareto-optimal Strategy Profiles}\label{sec:pareto}
We analyze the manipulation problem in the independent recruitment program of China's universities. In fact, this is also an open problem raised by \citet{Gusfield:1989} on what can be achieved by permutation manipulations. 
Formally, we have the following results in this section.
\begin{theorem}\label{theo:main_permutate}
	There exists a polynomial time algorithm (Algorithm~\ref{algorithm}) that, given any complete preference profile $P$ and any set of manipulators $L\subseteq W$ as input, computes a strategy profile $P'(L)$ such that when $L$ reports $P'(L)$, the induced matching $\mu'$ is $S_L$-Pareto-optimal.\footnote{It is weakly better off for all manipulators to follow the strategy $P'(L)$ rather than $P$, since $\mu'$ is stable under $P$, which is preferred by each manipulator to the W-pessimal matching under $P$.}
\end{theorem}

We will provide our algorithms and proof ideas behind Theorem \ref{theo:main_permutate} in Section~\ref{sec:algo}. Moreover, our algorithm provides an algorithmic characterization of $S_L$-Pareto-optimal matchings.
\begin{theorem}\label{theo:FindAllPareto}
	A matching is $S_L$-Pareto-optimal if and only if it is an induced matching of a strategy profile found by Algorithm \ref{algorithm}.
\end{theorem}

The full proofs of this section are deferred to \ref{app:proof_algo}.

\subsection{Conflicts between Manipulators}
Before we develop our algorithms, we first show an example to demonstrate that a coalition of women could get worse off if they perform their optimal single-agent manipulation separately.

\begin{table}[!htbp]
    \centering
	\begin{tabular}{c|cccccc|cccc}
		\cline{1-5} \cline{7-11}
		$m_1$ & $w_1$ & $w_4$ & $w_2$ & $w_3$  & \hspace{5em} & $w_1$ & $m_3$ & $m_2$ & $m_1$ & $m_4$  \\ \cline{1-5} \cline{7-11} 
		$m_2$ & $w_1$ & $w_3$ & $w_2$ & $w_4$  & \hspace{5em} & $w_2$ &  $m_1$ & $m_4$ & $m_3$ & $m_2$  \\ \cline{1-5} \cline{7-11} 
		$m_3$ & $w_2$ & $w_3$ & $w_1$ & $w_4$  & \hspace{5em} & $w_3$ &  $m_2$ & $m_3$ & $m_1$ & $m_4$  \\ \cline{1-5} \cline{7-11} 
		$m_4$ & $w_2$ & $w_4$ & $w_1$ & $w_3$  & \hspace{5em} & $w_4$ &  $m_4$ & $m_1$ & $m_3$ & $m_2$  \\ \cline{1-5} \cline{7-11}
		\multicolumn{5}{c}{Men's preferences} &  & \multicolumn{5}{c}{Women's preferences}          
	\end{tabular}

	\caption{Example of non-cooperativeness.}
	\label{hardexample}
\end{table}
		
Consider the preference lists in Table~\ref{hardexample}. The M-optimal matching is $\{(m_1, w_4),$ $(m_2, w_1), (m_3, w_3), (m_4, w_2)\}$. Suppose $L = \{w_1, w_2\}$ and consider individual manipulations by $w_1$ and $w_2$.
\begin{enumerate}
	\item $w_1$ exchanges $m_1$ and $m_2$ and get
	$\{(m_1, w_4), (m_2, w_3), (m_3, w_1), (m_4, w_2)\}$;
	\item $w_2$ exchanges $m_3$ and $m_4$ and get
	$\{(m_1, w_2), (m_2, w_1), (m_3, w_3), (m_4, w_4)\}$.
\end{enumerate} 
In both cases, $w_1$ and $w_2$ can manipulate to get their W-optimal partner and these manipulations are their optimal single-agent manipulation. However, if they jointly perform their optimal single-agent manipulations, the induced matching is $(m_1, w_1), (m_2, w_3), (m_3, w_2),$ $(m_4, w_4)$. It is surprising that they both get worse off than the matching resulted from their true preference lists. 

This example shows a sharp contrast between permutation manipulations and general manipulations, where removing men from the preference lists is allowed. In general manipulations, women can jointly perform their optimal single-agent manipulations to be matched with their W-optimal partner \citep{teo2001gale, gale1985ms}.

\subsection{Our Algorithm}
\label{sec:algo}
To develop our algorithm, we extensively use two structures, rotations and modified suitor graphs, introduced in Section~\ref{subsec:rotation} and~\ref{subsec:suitor}, respectively. We further develop several new structures such as maximal rotations and principle sets to derive connections between the modified suitor graphs and permutation manipulations.

Notice that eliminating more rotations results in 
weakly better matchings for all women.
Thus, the manipulators' objective is to eliminate as many rotations as possible by permuting their preference lists.
Since there is no direct rotation elimination in the Gale-Shapley algorithm, we try to figure out what kind of rotations can be eliminated, i.e., after eliminating these rotations, the corresponding matching is in $S_L$. 

We first analyze the structure of the sets of rotations. Rotations are not always exposed in a reduced table. Some rotations become exposed only after other rotations are eliminated. Thus, we define the {\em precedence relation} between rotations and based on that, we incorporate notions from \cite{irving1986complexity} (closed set, maximal rotations), and introduce the concept {\em principle sets} to analyze the problem.

The high-level idea behind our algorithm is that, with our theoretical analysis, we can reduce the search space from the set of all closed sets to the set of all principle sets, which enables our algorithm to run in polynomial time.

\begin{definition}[Precedence] \label{Precedence}
	A rotation $R_1=(\mathcal{M}_1, \mathcal{W}_1, \mathcal{W}_1^r)$ is said to explicitly precede another $R_2=(\mathcal{M}_2, \mathcal{W}_2, \mathcal{W}_2^r)$ if $R_1$ and $R_2$ share a common man $m$ such that $R_1$ moves $m$ from some woman to $w$ and $R_2$ moves $m$ from $w$ to some other woman. Let the relation precede be the transitive closure of the explicit precedence relation, denoted by $\prec$. Also, $R_1 \sim R_2$ if neither $R_1 \prec R_2$ nor $R_2 \prec R_1$.
\end{definition}

\begin{definition}[Closed set] \label{ClosedSets}
	A set of rotations $\mathcal{R}$ is closed if for each $R \in \mathcal{R}$, any rotation $R'$ with $R' \prec R$ is also in $\mathcal{R}$. A closed set $\mathcal C$ is minimal in a family of closed sets $\mathscr{C}$, if there is no other closed set in $\mathscr{C}$ that is a subset of $\mathcal C$. Moreover, define $CloSet(\mathcal{R})$ to be the minimal closed set that contains $\mathcal{R}$. 
\end{definition}

\begin{definition}[Maximal rotation \& Principle set] \label{MaximalRotation}
	Given a closed set of rotations $\mathcal{R}$, $R$ is a \emph{maximal rotation} of $\mathcal{R}$ if no rotation $ R' \in \mathcal{R}$ satisfies $R \prec R'$. Let $Max(\mathcal{R})$ be the set of all the maximal rotations in $\mathcal{R}$. Furthermore, $\mathcal{R}$ is a \emph{principle set} if $Max(\mathcal{R})$ contains only one rotation. We will slightly abuse notations and write $CloSet(R)$ to mean the principle set $CloSet(\mathcal{R})$ if $Max(\mathcal{R})=\{R\}$.
\end{definition}

Henceforth, $R_1$ precedes $R_2$ if $R_2$ can only be exposed after $R_1$ is eliminated. A rotation $R$ can only be exposed after all rotations preceding $R$ are eliminated. Thus only closed sets can be validly eliminated. Also, a closed set of rotations $\mathcal{R}$ is uniquely determined by $Max(\mathcal{R})$. Therefore, given a closed set $\mathcal{R}$, the corresponding matching after eliminating rotations in $\mathcal{R}$ is determined by $Max(\mathcal{R})$. 

The following theorem shows that closed sets of rotations are all that we need to consider.

\begin{theorem}[\citet{irving1986complexity}] \label{Correspondence}
	Let $S$ be the set of all stable matchings for a given preference profile, there is a one-to-one correspondence between $S$ and the family of all closed sets.
\end{theorem}

Therefore, we need to understand the changes made to the modified suitor graph when a rotation $R$ is eliminated. We keep track of every proposal made by men in $R$ and modify the graph accordingly. We first assume that the virtual vertex $s$ is comparable with each man and for every $w \in W$ and every $m \in M$, $m \succ_w s$. When eliminating a rotation, we follow the steps below to modify the graph:
\begin{enumerate}
    \item Let all women in $R$ reject their current partners, i.e., delete the edge $(w_i, m_i)$ involved in $R$ for each $i$;
    \item Arbitrarily choose a man $m_i$ in $R$ with no incoming edge from any woman and let him propose to the next woman $w$ in his preference list:
    \begin{enumerate}
        \item If $w$ is a manipulator, add an edge from $m_i$ to $w$ and delete edge $(s, w)$ if it exists;
        \item If $w$ is not a manipulator, then compare $m_i$ with the two men (one is possibly $s$) in $V'=\{v ~|~ (v, w) \in E\}$. If $m_i$ is not the worst choice, add an edge from $m_i$ to $w$ and delete the worst edge, and we say $w$ is \emph{overtaken} by $m_i$;
        \item If $w$ accepts $m_i$, add an edge from $w$ to $m_i$;
    \end{enumerate}
    \item Repeat step 2 until all men in $R$ are accepted.
\end{enumerate}

Let $G$ and $G'$ be the modified suitor graphs corresponding to the reduced tables before and after the elimination of $R$. It is easy to check that after changing $G$ using the steps defined above, the resulting graph is exactly $G'$. From Lemma \ref{CharOfProfile}, the most important property of the graph is the existence of a path from $s$ to any other vertex. Therefore, we focus on the change of {\em strongly connected components} and their connectivity in the modified suitor graph before and after the elimination of rotations.

\begin{definition} \label{StrongComponent}
	A sub-graph $G'$ is strongly connected if for any two vertices $u, v$ in $G'$, there is a path from $u$ to $v$ in $G'$. A strongly connected component is a maximal strongly connected sub-graph.
\end{definition}

The following lemma provides connectivity properties of the modified suitor graph after eliminating a rotation.

\begin{lemma} \label{RotationReachability}
    After eliminating a rotation $R$,
    \begin{enumerate}
        \item all agents in $R$ are in the same strongly connected component;
        \item vertices formerly reachable from a vertex in $R$ remain reachable from $R$;
        \item vertices overtaken during the elimination of $R$ are reachable from $R$.
    \end{enumerate}
\end{lemma}

With Lemma~\ref{RotationReachability}, we do not need to worry about vertices that are reachable from vertices in $R$, for they will remain reachable after the elimination. Also, vertices that are overtaken and the other vertices reachable from overtaken vertices can be reached from vertices in $R$ after the elimination. 

In fact, every vertex is reachable from $s$ in the initial graph. Therefore, if a vertex becomes unreachable from $s$ after eliminating a rotation, there must exist some edge that is deleted during the elimination, which only happens when some woman is overtaken if she is a non-manipulator. 
The next lemma extends Lemma~\ref{RotationReachability} to a closed set of rotations.

\begin{lemma} \label{ClosedSetsReachability}
	After eliminating a closed set of rotations $\mathcal{R}$, each $v$ in $\mathcal{R}$ is reachable from at least one vertex in $Max(\mathcal{R})$, i.e., there exists a path to $v$ from a vertex in $Max(\mathcal{R})$.
\end{lemma}

If $R_2$ explicitly precedes $R_1$, then they must contain a common man. Therefore, after eliminating $R_1$, vertices in $R_1$ can reach any vertex that is previously reachable from $R_2$. The analysis goes on recursively until some rotation has no predecessors.

Given a closed set of rotations $\mathcal{R}$, we say $\mathcal{R}$ can be eliminated for simplicity, if the corresponding matching after eliminating rotations in $\mathcal{R}$ is in $S_L$. The following lemma provides us a simpler way to check whether a closed set of rotations can be eliminated.

\begin{lemma} \label{Eliminatable}
    A closed set of rotations $\mathcal{R}$ can be eliminated if and only if after eliminating $\mathcal{R}$, every vertex in $Max(\mathcal{R})$ can be reached from $s$.
\end{lemma}

However, we still cannot afford to enumerate all possible closed sets of rotations, whose number is exponential with respect to the number of women.
\begin{theorem} \label{CloSetEliminatable}
    Given a closed set of rotations $\mathcal{R}$, if $\mathcal{R}$ can be eliminated, then there exists a rotation $R \in \mathcal{R}$ such that $CloSet(R)$ can be eliminated.
\end{theorem}

The above theorem reduces the search space from {\em the set of closed sets} to {\em the set of principles sets}. We are ready to design Algorithm \ref{algorithm} to compute a $S_L$-Pareto-optimal strategy profile. For any iteration of Algorithm~\ref{algorithm}, the matching at the beginning of each iteration is in $S_L$. Therefore, according to Theorem \ref{CloSetEliminatable}, if a closed set of rotations $\mathcal{R}$ can be eliminated, we can always find a principle set $\mathcal{P}^\ast$ contained in $\mathcal{R}$ such that $\mathcal{P}^\ast$ can be eliminated. 
Since the number of principle sets equals the number of rotations, which is polynomial and can be efficiently computed \citep{gusfield1987three}, given a matching in $S_L$, we figure out an efficient way to find a weakly better matching in $S_L$. Using this method as a sub-routine, we are able to design an algorithm to find a $S_L$-Pareto-optimal strategy profile for permutation manipulations. 

\begin{algorithm2e} 
	\caption{Find a $S_L$-Pareto-optimal strategy profile}\label{algorithm}
	Find the set of all rotations $\mathcal{R}$ and all principle sets $\mathscr{P}=\{CloSet(R)~|~R\in \mathcal{R}\}$;\\
	\While{True}{
		Construct $\mathscr C=\{\mathcal{P} \in \mathscr{P} ~|~ \mathcal{P} \text{ can be eliminated} \}$;\\
		\eIf{$\mathscr C=\emptyset$}{
			Construct $P(L)$ for $L$ and return;\\
		}{
			Arbitrarily choose a principle set $\mathcal{P}^\ast\in \mathscr C$ and eliminate $\mathcal{P}^\ast$;\\
		}
	}
\end{algorithm2e}

\subsubsection{Correctness of Algorithm~\ref{algorithm}}
To analyze the algorithm, we first consider the following lemma. 
\begin{lemma} \label{SubProblem}
	Given a set of manipulators $L \in W$, and the true preference profile $P=(P(M), P(W))$. Let $\mu$ be any matching in $S_L$ and $\mathcal{R}$ be the corresponding closed set of rotations. Then there exists a preference profile $P_{\mu}(L)$ for $L$ such that $\mu$ is the M-optimal stable matching of the preference profile $P_{\mu}=(P(M), P(N), P_{\mu}(L))$, and the reduced table of $P$ after eliminating $\mathcal{R}$ is exactly the reduced table of $P_{\mu}$ before eliminating any rotation.
\end{lemma}

Although from Theorem~\ref{Correspondence}, we know that only closed sets need to be considered,
there are still exponentially many possibilities. However, Theorem~\ref{CloSetEliminatable} shows that every closed set that can be eliminated contains a principle set, which can also be eliminated. A natural idea is to iteratively grow the closed set by adding principle sets. The above lemma shows that after each iteration, we can construct a problem that has the current matching as its initial matching, and contains rotations that are not yet eliminated. If we find a principle set that can be eliminated in the constructed problem, it can also be eliminated in the original problem.

\subsubsection{Complexity of Algorithm~\ref{algorithm}}
    To analyze the time complexity of Algorithm \ref{algorithm}, we define a graph describing the precedence relation between rotations.

\begin{definition}[Precedence graph] \label{PrecedenceGraph}
	Given a set of rotations $\mathcal{R}$, let $D$ be a directed acyclic graph, where the vertices in $D$ are exactly $\mathcal{R}$, and there is an edge $(R_1, R_2)$ in $D$ if $R_1 \prec R_2$. Moreover, let $H$ be the transitive reduction of $D$ defined above, and $H_r$ be the graph $H$ with all edges reversed.
\end{definition}

Note that $H$ is exactly the directed version Hasse diagram of the precedence relation between rotations. 
For a rotation $R$, $CloSet(R)$ is the set of vertices that can be reached from $R$ through a directed path in $H_r$. We split the algorithm into the initialization part and iteration part, and assume $|M|=|W|=n$. 

In the initialization part, we first compute the initial matching using the Gale-Shapley algorithm, which can be computed in $O(n^2)$ time. Next we find all rotations with respect to preference profile $P$ and also find all the principle sets. These two operations depend on the graph $H_r$. However, the graph $H$ is the transitive reduction of $D$, and the construction of $H$ is somewhat complex. \citet{gusfield1987three} discusses how to find all rotations, whose number is $O(n^2)$, in $O(n^2)$ time. Instead of constructing $H$, Gusfield considered a sub-graph $H'$ of $D$, whose transitive closure is identical to $D$. Moreover, $H'$ can be constructed in $O(n^2)$ time. We will not discuss how to construct $H'$ in detail but only apply Gusfield's results here. Then for each rotation $R$, we only need to search $H'$ to find $CloSet(R)$, which takes $O(n^2)$ time. Thus, we finish the initialization step in $O(n^4)$ time since there are $O(n^2)$ rotations altogether.

The iteration part is the bottleneck of the algorithm. At least one rotation is eliminated for each iteration, and thus $O(n^2)$ iterations are needed. Inside each iteration, we need to construct the set $\mathscr C$. There are $O(n^2)$ principle sets and to determine whether a principle set can be eliminated, we need to simulate the Gale-Shapley algorithm and change the modified suitor graph accordingly. After the modification, we traverse the graph to see if each vertex is reachable. Both of the two operations takes $O(n^2)$ time. Thus, the construction of $\mathscr C$ takes $O(n^4)$ time. In the \emph{If-Else} statement, if we find a principle set that can be eliminated, we eliminate the principle set and change the graph in $O(n^2)$. Otherwise, we traverse the graph to construct the preference profile for $L$ according to Lemma \ref{CharOfProfile}. Thus, the \emph{If-Else} statement takes $O(n^2)$ time and the time complexity of the algorithm is $O(n^6)$ in total.

\subsection{Algorithmic Characterization}
\label{sec:algo_char}
Notice that at each iteration, the algorithm has multiple principle sets to select from. To prove our characterization result in Theorem \ref{theo:FindAllPareto}, we have already shown that for each $S_L$-Pareto-optimal matching $\mu$, there exists a way to select the principle sets to eliminate in each iteration such that the induced matching from the output of Algorithm \ref{algorithm} is $\mu$.

\section{Inconspicuousness}
\label{sec:inconspicuousness}
In fact, if a stable matching with respect to the true preference lists can be obtained by permutation manipulations, the manipulators can also obtain this matching by an inconspicuous manipulation. We defer the proofs in this section to \ref{sec:proof_inconsp}.

\begin{definition}[Inconspicuous Strategy Profile]
    A strategy profile is inconspicuous if each manipulator permutes their preference lists by moving only one man to a higher rank. 
\end{definition}

For convenience, we introduce a new notation $Pro(w)$ for each $w \in W$. A proposal list $Pro(w)$ of a woman is a list of all men who have proposed to her in the Gale-Shapley algorithm, and the orderings of its entries are the same as her stated preference list. A reduced proposal list contains the top two entries (first entry if there is only one entry) of $Pro(w)$, denoted by $Pro_r(w)$. Clearly, each woman $w$ is matched to the first man of $Pro_r(w)$.

\begin{lemma} \label{ReducedPropList}
	Given all agents' true preference profile $(P(M), P(W))$, if a matching $\mu$ is in $S_L$ with corresponding preference profile $P = (P(M), P(N), P(L))$, then the induced matching is still $\mu$, if for each $w \in L$, we modify $w$'s preference list by moving $Pro_r(w)$ to the top and ordering other men arbitrarily.
\end{lemma}

\begin{theorem}[Inconspicuous manipulation]\label{TinyChange}
For any stable matching with respect to the true preference lists that can be obtained by permutation manipulations, there exists a preference profile for the manipulators, in which each manipulator only needs to move at most one man to some higher ranking, that yields the same matching.
\end{theorem}

Theorem~\ref{TinyChange} suggests that for each woman $w$, let $m_1$ and $m_2$ be the two men in $Pro_r(w)$\footnote{If woman $w$ only receives one proposal, she cannot implement any manipulation.}, then $w$ can modify her true preference list by moving $m_2$ to the place right after $m_1$ to generate the same induced matching (see Algorithm~\ref{algorithm2} for details). 


\begin{algorithm2e}
	\caption{Find a $S_L$-Pareto-optimal and inconspicuous preference profile}
	\label{algorithm2}
	Use Algorithm \ref{algorithm} to compute a strategy profile $P'(L)$ for $L$;\\
	Compute $Pro_r(w)$ for each $w\in L$ with respect to $P'(L)$;\\
	\For{$w$ in $L$}{
		Modify the true preference list $P(w)$ by moving the second man in $Pro_r(w)$ to the position right after the first man in $Pro_r(w)$;\\
	}
	\Return{ the modified preference profile $P$};
	
\end{algorithm2e}

\section{Incentive Properties}
\label{sec:incentive}
Although we only have been focusing on constructing $S_L$-Pareto-optimal strategy profiles, a $S_L$-Pareto-optimal strategy profile, which is also inconspicuous, actually forms a Nash equilibrium.

\begin{lemma}
\label{lem:opt_single_manip}
	Suppose there is only one manipulator $w$. Then the best matching $\mu'$ that $w$ can obtain via permutation manipulation is stable with respect to the true preference $P$.
\end{lemma}


\begin{theorem}\label{thm:charOfNE}
    A strategy profile, that is $S_L$-Pareto-optimal and inconspicuous, forms a Nash equilibrium. 
\end{theorem}
\begin{proof}
    Denote by $P$ and $\mu$ the true preference profile of agents and the corresponding matching. Let $P_1$ be the preference profile returned by Algorithm~\ref{algorithm2} given $P$, and $\mu_1$ be the corresponding matching. It is clear that for each $w\in L$, Algorithm~\ref{algorithm2} only changes the order of the men ranked strictly lower than $\mu_1(w)$. For the sake of contradiction, assume there exists a manipulator $w'\in L$ such that $w'$ can get a strictly better partner $m$ ($m\succ^P_{w'}\mu_1(w')$) by misreporting a different preference list. Let $P_2$ and $\mu_2$ be the preference profile after misreporting and the corresponding matching. Without loss of generality, we assume that $m$ is the best partner (according to both $P$ and $P_1$) that $w'$ can obtain. Then we know from Lemma~\ref{lem:opt_single_manip} that, $\mu_2$ is stable with respect to $P_1$, and thus for each $w\in W$, we have that $\mu_2(w)\succeq^{P_1}_w \mu_1(w)$. It follows that $\mu_2(w)\succeq^{P}_w \mu_1(w)$, since Algorithm~\ref{algorithm2} does not change the order of the men who are ranked higher than $\mu_1(w)$. It follows that $\mu_2$ is also stable with respect to $P$, and $\mu_2$ Pareto-dominates (in the sense of $S_L$-Pareto-optimality) $\mu_1$. However, $\mu_2$ is not found by Algorithm~\ref{algorithm2}. A contradiction.
\end{proof}

Since all $S_L$-Pareto-optimal strategy profiles can be turned into an inconspicuous manipulation by Algorithm~\ref{algorithm2}, we have the following corollary.
\begin{corollary}
    For any $S_L$-Pareto-optimal matching, there exists a Nash equilibrium that can induce it.
\end{corollary}
Therefore, $S_L$-Pareto-optimal matchings exactly address both the cooperation and the competition among the women in the coalition.

\section{Manipulations in the many-to-one setting}
\label{sec:many}
We have already discussed the manipulation problem in the one-to-one setting. However, China's college admissions process is a many-to-one setting, since each university can be matched with multiple students. It is worth mentioning that there is almost no tie between students. If two students have the same total scores, the admissions process breaks ties by comparing the their scores of different subjects sequentially. In this section, we analyze the manipulation problem in the many-to-one setting. In fact, all main results presented in previous sections can be naturally extended to the many-to-one setting.

Let $q_i\ge 1$ be the quota of $w_i$, i.e., $w_i$ can be matched to at most $q_i$ men. In China, the quota of a university is almost fixed and are always publicly known (changes need to be approved by the Ministry of Education). Therefore, we assume that $q_i$ is known to every agent and do not consider manipulations by misreporting the quota.

Before we start our analysis, we need to emphasize that the women's preferences are slightly different in this setting. We need to define preferences between sets of men, since each woman has a quota more than 1. We first define the \emph{responsive preference relation}.

\begin{definition}[Responsive preference relation]
	Suppose that $w$ has strict preferences over individual men $\prec_w$. A preference relation over sets of men is responsive if, $w$ prefers set $S$ to $T$, whenever $S$ and $T$ satisfy:
	\begin{enumerate}
		\item $m\in S$, $m'\in T$ and $S = T\cup \{m\} \setminus\{m'\}$;
		\item $m\succ_w m'$.
	\end{enumerate}
\end{definition}

We assume that the women's preferences are the transitive closure of the above responsive preference relation. Or equivalently,

\begin{definition}[Set preferences of women]
$w$ prefers set $S$ to $T$ if there exists a one-to-one mapping $\pi:S\mapsto T$, such that $\forall m\in S$, $m\succ_w \pi(m)$, where $\succ_w$ is the individual preference.
\end{definition}

Note that this only defines relations over sets that have the same size. We focus on matchings that are stable with respect to true preferences and it is known that in all stable matchings, each woman is matched with the same number of men. Therefore such a definition is enough for our analysis.

We omit the formal definition of Nash equilibrium and $S_L$-Pareto-optimality, since they can be easily adapted with the above definition to suit the many-to-one setting. For ease of presentation, we will use \emph{set Nash equilibrium} and \emph{set $S_L$-Pareto-optimality} to mean the corresponding definitions in terms of the set preferences of women. To apply the Gale-Shapley algorithm to this setting, one can simply think of each woman $w_i$ to be $q_i$ copies of the same woman (with the same preference list), each with a quota of 1. We also need to change men's preference lists by replacing each $w_i$ with the $q_i$ copies $w_{i,1}, w_{i,2}, \dots, w_{i,q_i}$. For simplicity, we assume that the $q_i$ copies are always placed in this order. Since these copies correspond to the same $w_i$, this assumption is without loss of generality. We call the new instance \emph{the corresponding one-to-one instance}. 

We can then apply all our previous results to the corresponding one-to-one instance. Suppose that we run Algorithm \ref{algorithm} on it and get a $S_L$-Pareto-optimal (in terms of individual preferences) matching. It is easy to see that this matching is also set $S_L$-Pareto-optimal. But in the many-to-one setting, we have an additional constraint: all $q_i$ copies of a manipulator $w_i$ must have the same preference list, since they actually represent the same woman. 

\begin{theorem}[Inconspicuous manipulation in the many-to-one setting]\label{thm:inc_many}
	Given a many-to-one instance, for any $S_L$-Pareto-optimal matching $\mu$ computed by Algorithm \ref{algorithm} on the corresponding one-to-one instance, the same matching $\mu$ can be achieved on the many-to-one instance and each manipulator $i$ moves at most $q_i$ men to some higher rankings.
\end{theorem}

The proof is deferred to \ref{sec:pf_opt_single_manip}.

\section{Strictly Better-off Outcomes}\label{sec:strict_better}

The above results show that the Gale-Shapley algorithm is vulnerable to coalition manipulation. However, under the setting where a manipulation is costly, every manipulator needs to be strictly better off after the manipulation to preserve individual rationality. In the example that demonstrates the conflicts between manipulators in Table~\ref{hardexample}, only one of $w_1$ and $w_2$ can manipulate or the $W$-optimal partner. Therefore, such an example provides a natural way to represent a binary variable. In fact, we show a hardness result in the costly environment:

    \begin{theorem}
    	\label{theo:hardness_better}
        It is NP-complete to find a strategy profile, the induced matching of which is strictly better off for all manipulators.
    \end{theorem} 
    
    Therefore, when the manipulation is costly, a manipulation coalition is unlikely to form and the Gale-Shapley algorithm is immune to coalition manipulations. According to Theorem \ref{theo:hardness_better}, one immediate corollary is that the number of $S_L$-Pareto-optimal matchings cannot be polynomial in the number of men and women. For otherwise, we can enumerate all such matchings by Algorithm~\ref{algorithm} to develop a polynomial time algorithm. 
	Last but not least, we show that the problem to compute the number of $S_L$-Pareto-optimal matchings, which are strictly better off for all manipulators, is \#P-complete.	

Finally, we show that computing the number of $S_L$-Pareto-optimal matchings which are strictly better off for all manipulators  is \#P-Hard.
    \begin{theorem}\label{sharp_P}
        It is \#P-complete to compute the number of $S_L$-Pareto-optimal matchings, which are strictly better off for all manipulators.    
    \end{theorem}

\section{Conclusion}
Motivated by a real life phenomenon risen in recent years in the college admissions process in China, we consider manipulations by a subset of women in the Gale-Shapley algorithm. We show that a Nash equilibrium with $S_L$-Pareto-optimal matching can be efficiently computed in general. These results confirm that the leagues of universities can benefit from forming coalitions. On the contrary, we show that it is NP-complete to find a strictly better off matching for all the manipulators, implying that Gale-Shapley algorithm is immune from permutation manipulations when the manipulations are costly.

\section{Acknowledgements}
Weiran Shen is with Beijing Key Laboratory of Big Data Management and Analysis Methods, Gaoling School of Artificial Intelligence, Renmin University of China. This work was supported by Beijing Outstanding Young Scientist Program NO. BJJWZYJH012019100020098 and Intelligent Social Governance Interdisciplinary Platform, Major Innovation \& Planning Interdisciplinary Platform for the ``Double-First Class'' Initiative, Renmin University of China.

\bibliographystyle{plainnat}
\bibliography{reference}
\clearpage

\appendix
\section*{APPENDIX}
\setcounter{section}{0}

\section{Differences in the Definitions of Suitor Graphs}
\label{app:suitor_graph}
The definition of suitor graphs was initially proposed in \cite{kobayashi2009successful}, but the name ``suitor graph'' is used in their subsequent paper \cite{kobayashi2010cheating}. However, the setting of \citet{kobayashi2009successful} does not consider the problem of manipulations, hence there is no manipulators or non-manipulators. Therefore, compared with the original suitor graph (See Definition \ref{def:suitor_graph_2009} below), our definition adds edges for $w \in L$  instead of for all women and contains an additional Step 4. Our definition also includes the virtual vertex, which is only introduced later to prove their main result in \cite{kobayashi2009successful}.
\begin{definition}[Suitor Graph \cite{kobayashi2009successful}]
\label{def:suitor_graph_2009}
Given all men's preference profile $P(M)$ and a matching $\mu$, the corresponding suitor graph consists of:
\begin{enumerate}
	\item a set of vertices: $M\cup W$;
	\item a set of directed edges:
	\begin{align*}
		&\left\{ (w, \mu(w))\in W\times M \mid w\in W \right\} \\
		\cup &\left\{(m, w)\in M\times W \mid w=\mu(w) \text{ or } w\succ_m \mu(m)  \right\}
	\end{align*}
\end{enumerate}
\end{definition}

\citeauthor{kobayashi2009successful} later defined the suitor graph in a slightly different way in their subsequent paper \cite{kobayashi2010cheating}, where they remove the edges $(\mu(w), w)$ for all $w$. They also define the \emph{rooted suitor graph}, which includes the same virtual vertex as in our definition. 

\begin{definition}[Suitor Graph \cite{kobayashi2010cheating}]
	\label{def:suitor_graph_2010}
	Given all men's preference profile $P(M)$ and a partial matching $\mu$ that only specifies the partners of a subset of the agents, the corresponding suitor graph consists of:
	\begin{enumerate}
		\item a set of vertices: $M\cup W$;
		\item a set of directed edges:
		\begin{align*}
			&\left\{ (w, \mu(w))\in W\times M \mid w \text{ is matched in } \mu \right\} \\
			\cup &\left\{(m, w)\in M\times W \mid m \text{ is matched in } \mu \text{ and } w\succ_m \mu(m)  \right\}
		\end{align*}
	\end{enumerate}
\end{definition}

\section{Omitted Proofs in Section \ref{sec:pareto}.}
\label{app:proof_algo}
\subsection{Proof of Lemma \ref{RotationReachability}.}
\begin{replemma}{RotationReachability}
After eliminating a rotation $R$,
\begin{enumerate}
    \item all agents in $R$ are in the same strongly connected component;
    \item vertices formerly reachable from a vertex in $R$ remain reachable from $R$;
    \item vertices overtaken during the elimination of $R$ are reachable from $R$.
\end{enumerate}
\end{replemma}

To prove Lemma \ref{RotationReachability}, we first show the following claim.
\begin{claim} \label{FirstAcceptance}
	For each man $m_i$ in $R$, in the procedure of eliminating the rotation $R$, $w_{i+1}$ (the subscript is taken modulo $r$) is the first woman to accept him, and each woman in $R$ accepts only one proposal during the procedure.
\end{claim}
\begin{proof}
According to the definition of rotations, $w_{i+1}$ is the second in $m_i$'s reduced list. If there are other women between $w_i$ and $w_{i+1}$ in $m_i$'s preference list, they are absent from the reduced list because these women already hold proposals from better men. Henceforth, even though $m_i$ proposes to these women, they reject him. But $m_i$ is in $w_{i+1}$'s reduced list since $w_{i+1}$ is in $m_i$'s. Therefore, $m_i$ is a better choice for $w_{i+1}$ and $w_{i+1}$ accepts him. 

After the elimination, each man $m_i$ in $R$ proposes to $w_{i+1}$ and each man is accepted only once. Also each woman $w_{i+1}$ holds a new proposal from $m_i$ and thus accepts at least once. The conclusion is immediate since the total number of accepted
men is equal to the total number of women who accept a new partner.
\end{proof}

\begin{proof}[Proof of Lemma \ref{RotationReachability}.]
For each $m_i$ in $R$, $R$ moves $m_i$ from $w_i$ to $w_{i+1}$. As a result, there exists an edge from $w_{i+1}$ to $m_i$. We now prove that each $m_i$ has an outgoing edge pointing to $w_i$, and all agents in $R$ then form a cycle, and thus in the same strongly connected component. Before the elimination, $w_i$ is the partner of $m_i$, so there is an edge from $m_i$ to $w_i$. If $w_i$ is a manipulator, the edge $(m_i, w_i)$ is not removed during the elimination according to the steps described above. If $w_i$ is not a manipulator, then only two incoming edges are remained after the elimination and these edges are from the best two men among those who propose to her. According to Claim~\ref{FirstAcceptance}, only one man, namely $m_{i-1}$, is accepted. Thus, $m_{i-1}$ is the best suitor of $w_i$. We claim that $m_i$ is the second best and the edge from $m_i$ is still in the modified suitor graph. Otherwise, suppose $m'$ is a better choice than $m_i$ to $w_i$. Then $m'$ is also in $R$. We let $m'$ propose first, and $w_i$ accepts $m'$, which makes $w_i$ accepts at least twice. A contradiction.

Since each woman can be reached from her partner before the elimination, it is without loss of generality to assume that a vertex $v$ can be reached from a man $m$ in $R$ through a path $p$. Let $u$ be the last vertex in $p$ such that $u$ is in $R$ or is overtaken by a vertex in $R$. If $u$ is in $R$, then after the elimination, $m$ can reach $u$ since they are in the same strongly connected component. If $u$ is overtaken by some vertex $m'$, then during the elimination, an edge $(m', u)$ is added to the graph. Thus, $m$ can reach $u$ through $m'$. Henceforth, in any case, $u$ is reachable. Since in $p$ the vertices between $u$ and $v$ are neither in $R$ nor overtaken by some vertex in $R$, the path from $u$ to $v$ remains in the modified graph. Therefore $v$ is reachable from $m$ and also from any vertex in $R$ for they are in the same strongly connected component after the elimination.
\end{proof}

\subsection{Proof of Lemma \ref{ClosedSetsReachability}}

\begin{replemma}{ClosedSetsReachability}
	After eliminating a closed set of rotations $\mathcal{R}$, each $v$ in $\mathcal{R}$ is reachable from at least one vertex in $Max(\mathcal{R})$, i.e., there exists a path to $v$ from a vertex in $Max(\mathcal{R})$.
\end{replemma}

\begin{proof}
We eliminate the rotations in $\mathcal{R}$ one by one and generate a sequence of rotations $q=(R_1, R_2, \ldots, R_n)$. $R_i$ is the $i$-th rotation to eliminate. After eliminating $R_n$, all rotations in $\mathcal{R}$ are eliminated. Denote $q_i = \bigcup_{j=1}^i R_j$. For each $i$, $q_i$ is a closed set. We call $i$ the sequence number of $q_i$ and we prove by induction on the sequence number that after eliminating $q_i$, all vertices in $q_i$ can be reached from a vertex in $Max(q_i)$. For $i=1$, $q_i=\{R_1\}$, the case is trivial from Lemma~\ref{RotationReachability}. Assume the statement is true for $i=k$, then for $i=k+1$, we only eliminate one more rotation $R_{k+1}$ than in the case with $i=k$. $R_{k+1}$ is in $Max(q_{k+1})$ otherwise there exists another rotation $R'$ in $q_k$ such that $R_{k+1} \prec R'$ and then $q_k$ is not a closed set. Let $D=Max(q_k)\setminus Max(q_{k+1})$. Rotations in $D$ are no longer maximal rotations because $R_{k+1}$ is eliminated, which indicates that rotations in $D$ explicitly precede $R_{k+1}$. Henceforth, every rotation $R$ in $D$ has a common agent with $R_{k+1}$ and each vertex $u$ reachable from $R$ is reachable from that common agent. According to Lemma~\ref{RotationReachability}, $u$ can be reached from $R_{k+1}$. For each vertex $u'$ that is not reachable from rotations in $D$, it must be reachable from another rotation $R'$ in $Max(q_k)$ through path $p$ and $R'$ is still in $Max(q_{k+1})$. If $p$ is still in the graph, then we are done. Otherwise, some vertices in $p$ must be in $R_{k+1}$ or overtaken by a man in $R_{k+1}$. Let $z$ be the last vertex in $p$ such that $z$ is in $R_{k+1}$ or overtaken. $z$ can be reached from $R_{k+1}$ and the path from $z$ to $u'$ is not affected by the elimination. Therefore, $u'$ is reachable from $R_{k+1}$.
\end{proof}

\subsection{Proof of Lemma \ref{Eliminatable}.}
\begin{replemma}{Eliminatable}
	A closed set of rotations $\mathcal{R}$ can be eliminated if and only if after eliminating $\mathcal{R}$, every vertex in $Max(\mathcal{R})$ can be reached from $s$.
\end{replemma}
\begin{proof}
If a closed set of rotations $\mathcal{R}$ can be eliminated, then every vertex is reachable after $\mathcal{R}$ is eliminated. As a result, any member of $Max(\mathcal{R})$ is reachable.

If after eliminating $\mathcal{R}$, any member of $Max(\mathcal{R})$ can be reached from $s$, then we need to show that all other vertices are also reachable from $s$. We split all vertices into two parts. Let $V$ denote the set of all the vertices that can be reached from members of $Max(\mathcal{R})$. If a vertex $v$ is in $V$, then $v$ is reachable from $s$ through $Max(\mathcal{R})$. If $v$ is not in $V$, then in the initial graph, there is a path $p$ from $s$ to $v$. We claim that no vertex in path $p$ is either in any of the rotations in $\mathcal{R}$ or overtaken when eliminating a rotation. Otherwise,  according to Lemma~\ref{ClosedSetsReachability}, $v$ is reachable from $Max(\mathcal{R})$. Thus, the path $p$ is still in the graph after eliminating all the rotations in $\mathcal{R}$.
\end{proof}

\subsection{Proof of Theorem \ref{CloSetEliminatable}.}
\begin{reptheorem}{CloSetEliminatable}
	Given a closed set of rotations $\mathcal{R}$, if $\mathcal{R}$ can be eliminated, then there exists a rotation $R \in \mathcal{R}$ such that $CloSet(R)$ can be eliminated.
\end{reptheorem}
In order to prove Theorem \ref{CloSetEliminatable}, we first show the following claim about the maximal rotations of a closed set that can be eliminated.
\begin{claim} \label{EliContainLiar}
	If a closed set $\mathcal{R}$ can be eliminated, then every rotation in $Max(\mathcal{R})$ must contain a manipulator.
\end{claim}
\begin{proof}
Assume there exists a rotation $R\in Max(\mathcal{R})$ such that $R$ contains no manipulator. We can change the order of elimination to make $R$ the last to eliminate. We prove that after eliminating $R$, any vertex in $R$ is not reachable from $s$. From the proof of Lemma~\ref{RotationReachability}, we know that all vertices in $R$ form a cycle after eliminating $R$. Each man in $R$ has only one incoming edge from his current partner who is also in $R$. Each woman has two incoming edges, one from her partner in $R$ and another from her former partner which is also in $R$. Thus, every vertex in $R$ has no incoming edges from outside the cycle and thus is not reachable from $s$.
\end{proof}
\begin{proof}[Proof of Theorem \ref{CloSetEliminatable}.]
Let $V$ be the set of all vertices in $\mathcal{R}$. After eliminating $\mathcal{R}$, we arbitrarily choose a vertex $v$ in $V$. In the corresponding modified suitor graph, there is a path $p=(v_0=s, v_1, v_2, \ldots, v_n=v)$ from $s$ to $v$ since $\mathcal{R}$ can be eliminated. Let $u$ be the first vertex in $p$ such that $u$ is in $V$. $u$ is obviously not $v_1$, or otherwise the edge $(s, u)$ will be deleted. Moreover, $u$ must be in $L$, since any non-manipulator can only be reached from a node in $V$ if she is overtaken during the elimination. Assume $u=v_l$ and $l > 1$. Then the sub-path $p'=(v_0, v_1, \ldots, v_l=u)$ is not affected (no vertices in $V$ or overtaken) during the elimination. Henceforth, $p'$ is in the original graph before eliminating $\mathcal{R}$. Now we consider the set $\mathcal{R}'=\{R \in \mathcal{R} | u \in R\}$. For any $R$ in $\mathcal{R}'$, if we eliminate $CloSet(R)$, the sub-path is also not affected. Therefore $CloSet(R)$ can be eliminated according to Lemma~\ref{Eliminatable}.
\end{proof}

\subsection{Proof of Lemma \ref{SubProblem}.}

\begin{replemma}{SubProblem}
	Given a set of manipulators $L \in W$, and the true preference profile $P=(P(M), P(W))$. Let $\mu$ be any matching in $S_L$ and $\mathcal{R}$ be the corresponding closed set of rotations. Then there exists a preference profile $P_{\mu}(L)$ for $L$ such that $\mu$ is the M-optimal stable matching of the preference profile $P_{\mu}=(P(M), P(N), P_{\mu}(L))$, and the reduced table of $P$ after eliminating $\mathcal{R}$ is exactly the reduced table of $P_{\mu}$ before eliminating any rotation.
\end{replemma}
\begin{proof}
Since $\mu$ is in $S_L$, there exists $P' = (P(M), P(N), P'(L))$ such that $\mu$ is the induced matching for $P'$. For each $w \in L$, we modify $P'(w)$ as follows:
\begin{enumerate}
	\item delete all men $m$ such that $m \succ_w^P \mu(w)$;
	\item reinsert them at the beginning according to their order in $w$'s true preference list;
	\item move $\mu(w)$ to the position right after all men $m$ such that $m \succ_w^P \mu(w)$.
\end{enumerate}
Denote the modified preference profile by $P'_\mu$. In fact, $P'_\mu$ is the $P_{\mu}$ we are looking for.

We first prove that $\mu$ is the M-optimal matching under $P'_\mu$. After the first two steps of modifications, the M-optimal matching is still $\mu$, since for each $w$, we only change the position of men ranked higher than $\mu(w)$ in her true preference list, who must have not proposed to $w$ under $P'$, and thus do not change the output of the Gale-Shapley algorithm. Otherwise, if a man $m$ with $m \succ_w^P \mu(w)$ has proposed to $w$, then we must have $w \succ_m^{P'} \mu(m)$, which is equivalent to $w \succ_m^{P} \mu(m)$. Thus $(m, w)$ forms a blocking pair in $\mu$ under the true preference profile $P$, contradicting to the stability of $\mu$ under $P$. In the third step, we move $\mu(w)$ to the position right after all men ranked higher than $\mu(w)$ in the true preference list $P(w)$. Consider all the men $m'$ with $m' \succ_w^{P'} \mu(w)$ but $\mu(w) \succ_w^{P'_\mu} m'$. $m'$ must have not proposed to $w$ under $P'$, or otherwise $\mu(w)$ cannot be the partner of $w$. Therefore, the positions of the men in $P'_\mu$ do not affect the output of the Gale-Shapley algorithm.

Let $T_{P_\mu}$ be the reduced table of $P$ after eliminating $\mathcal{R}$ and $T_{P'_\mu}$ be the reduced tables of $P'_\mu$. We already know that for each woman, her partners in the two reduced tables are the same, which is $\mu(w)$. In fact, a change of reduced table happens if and only if a woman accepts a proposal from a man $m$ and removes everyone ranked below $m$ in her preference list. So in the reduced list of each woman, no man is ranked below her current partner. Therefore, to prove that $T_{P_\mu}$ is the same as $T_{P'_\mu}$, it suffices to show that for each woman, $P$ and $P'_\mu$ are the same after removing all men ranked below her current partner, which is clear from the construction of $P'_\mu$.
\end{proof}

\subsection{Proof of Theorem \ref{theo:FindAllPareto}.}
\begin{reptheorem}{theo:FindAllPareto}
	A matching is $S_L$-Pareto-optimal if and only if it is an induced matching of a strategy profile found by Algorithm \ref{algorithm}.
\end{reptheorem}
\begin{proof}
Assume the $P(L)$ is a $S_L$-Pareto-optimal strategy profile for the manipulators. Let $\mu$ be the matching produced by $P(L)$ and $\mathcal{R}_{\mu}$ the corresponding set of rotations. $\mu$ can be forced to be the induced matching by always choosing the principle set that is a subset of $\mathcal{R}_{\mu}$. Let $\mathcal{R}_k$ be the rotations eliminated so far at the end of the $k$-th iteration and $\mu_k$ be the corresponding matching . We prove by induction on the iterations that at the end of each iteration, $\mathcal{R}_k$ is a subset of $\mathcal{R}_{\mu}$. In the first iteration, $\mathcal{R}_{\mu}$ is in $S_L$, so there exists a principle set $\mathcal{P} \subset \mathcal{R}_{\mu}$ that can be eliminated. Assume the statement holds for the $k$-th iteration. At the beginning of the $(k+1)$-th iteration, $\mu_k$ is the induced matching, and $\mathcal{R}_k$ is a subset of $\mathcal{R}_{\mu}$ by the inductive hypothesis, then there exists at least one principle set $\mathcal{P}_{k+1} \subset \mathcal{R}_{\mu} \setminus \mathcal{R}_k$ that can be eliminated. Therefore, at the end of the $(k+1)$-th iteration, $\mathcal{R}_{k+1} = \mathcal{R}_k \cup \mathcal{P}_{k+1}$ is also a subset of $\mathcal{R}_{\mu}$. When the algorithm terminates, the set of all eliminated rotations $\mathcal{R}$ is also a subset of $\mathcal{R}_{\mu}$. Assume $\mathcal{R} \ne \mathcal{R}_{\mu}$, then we can find some principle set to eliminate, which contradicts to the termination condition of the algorithm. Therefore the $S_L$-Pareto-optimal strategy profile can be found by the algorithm.
\end{proof}

\section{Omitted Proofs in Section \ref{sec:inconspicuousness}.}
\label{sec:proof_inconsp}

\subsection{Proof of Lemma \ref{ReducedPropList}.}
\label{sec:proof_reduced_proplist}
\begin{replemma}{ReducedPropList}
	Given all agents' true preference profile $(P(M), P(W))$, if a matching $\mu$ is in $S_L$ with corresponding preference profile $P = (P(M), P(N), P(L))$, then the induced matching is still $\mu$, if for each $w \in L$, we modify $w$'s preference list by moving $Pro_r(w)$ to the top and ordering other men arbitrarily.
\end{replemma}

\begin{proof}
Suppose the corresponding matching to the modified preference profile is $\mu'$. We show that $\mu' = \mu$. 

Let $P$ and $P'$ be the original profile and the modified profile. All the partial orders we used in this proof is defined in $P$. We construct a graph $T$, which is a sub-graph of the modified suitor graph $G(P(M), P(N), \mu)$, according to the set of all reduced proposal lists in $P$. The set of vertices is just $M \cup W$, and the edges are $E = \{(w, \mu(w))~|~w\in W\} \cup \{(m, w)~|~w \succ_m \mu(m), m \in Pro_r(w) \}$. We also add a virtual vertex $s$, and add edges from $s$ to each woman who has no incoming edges. Note that every woman has an outgoing edge pointing to her mate in $\mu$, and at most one incoming edge from her second entry in her proposal list. It is easy to prove that at least one woman has only one entry in her proposal list, and thus this woman has no incoming edge except the one from $s$. 

It is straightforward to check that $\mu$ is also stable under $P'$. Then we have $\mu'(m) \succeq_m \mu(m)$, which indicates that if $m$ proposes to some woman $w$ in $P'$, then he also proposes to her in $P$. Now we can prove the lemma by induction on the height of the breadth-first search tree on graph $T$ rooted at $s$. Denote the height of a vertex as $h(v)$. For each vertex with $h(v) = 1$, it must be a woman and has no incoming edge from vertices of $M$. Therefore, she gets only one proposal from $\mu(w)$ in $P$. Therefore each man $m$ other than $\mu(w)$ must be matched to a better woman, i.e., $\mu(m) \succ_m w$. Also, as proved above $\mu'(m) \succeq_m \mu(m)$. Then we have $\mu'(m) \succ w$, which means $m$ does not propose to $w$ in $P'$. The only possible partner for $w$ is $\mu(w)$. Thus, we can conclude that she is matched with $\mu(w)$ in $\mu'$, or $\mu'(w)=\mu(w)$. 

Assume $\mu'(v)=\mu(v)$ for each $v$ with $h(v) = k$, then for a vertex $v'$ with $h(v')=k+1$, we prove that we still have $\mu'(v')=\mu(v')$. If $k+1$ is even, then $v'$ is a man and we consider $v'$'s parent $v=Prt(v')$. From the construction of the graph, there is an edge from $v$ to $\mu(v)$. Henceforth, according to the inductive hypothesis, $\mu'(v)=\mu(v) = v'$, and $\mu(v') = v = \mu'(\mu'(v)) = \mu'(v')$. If $k+1$ is odd, then $v'$ is a woman and there is an edge pointing to her from $v$ who is the second entry in her received proposal list. On the one hand, each man in $\{m | m \succ_v \mu(v)\}$ is matched with someone who is better than $v$ in $\mu$. As a result, $\mu(m) \succ_m v$. And still $\mu'(m) \succeq_m \mu(m)$, we have $\mu'(m) \succ v$. Therefore $m$ does not propose to her in $P'$. On the other hand, $\mu(v)$ proposes to $v$ in $P'$ since $\mu(v)$ proposes to her in $P$. Combining the two sides, we know that $\mu(v)$ is the best man among all those who propose to her. Thus, $\mu'(v) = \mu(v)$.
\end{proof}

\subsection{Proof of Theorem~\ref{TinyChange}}
\begin{reptheorem}{TinyChange}[Inconspicuous manipulation]
    For any stable matching with respect to the true preference lists that can be obtained by permutation manipulations, there exists a preference profile for the manipulators, in which each manipulator only needs to move at most one man to some higher ranking, that yields the same matching.
\end{reptheorem}
\begin{proof}
We first construct the modified suitor graph using $\mu$ and compute the corresponding $P(L)$ according to Lemma \ref{CharOfProfile}. After that, we can compute $Pro(w)$ and $Pro_r(w)$ for each woman $w$ according to $P(L)$. Then we just move the second entry (if exists) of $Pro_r(w)$ to the position right after $\mu(w)$ in each manipulator $w$'s original preference list. Notice that in the modified preference list, no man who is ranked higher than $\mu(w)$ in $w$'s preference list proposes to $w$, or otherwise the induced matching is unstable with respect to true preference lists. Thus, the orderings of these men is irrelevant to the matching result and we can move $Pro_r(w)$ to the top without affecting the induced matching $\mu'$ for the modified lists. According to Lemma \ref{ReducedPropList}, we can conclude that $\mu' = \mu$.
\end{proof}

\section{Omitted Proofs in Section \ref{sec:incentive}.}
\subsection{Proof of Lemma \ref{lem:opt_single_manip}.}\label{sec:pf_opt_single_manip}
\begin{replemma}{lem:opt_single_manip}
	Suppose there is only one manipulator $w$. Then the best matching $\mu'$ that $w$ can obtain via permutation manipulation is stable with respect to the true preference $P$.
\end{replemma}
\begin{proof}
Let $P'$ be the preference profile corresponding to $\mu'$. Assume on the contrary that $\mu'$ is not stable with respect to $P$. Then there must be a blocking pair. However, any pair $(m, w')$ with $w'\ne w$ cannot block $\mu'$ under $P$, since they have the same preferences in both $P$ and $P'$. It follows that the woman in the blocking pair must be $w$. Let $(m, w)$ be the blocking pair. We move $m$ to the top of $P'(w)$. If we run the Gale-Shapley algorithm with the new preference profile, $m$ will still propose to $w$ and will finally be matched to $w$ since $m$ is now the favorite man of $w$. But $m\succ_w\mu'(w)$, which contradicts to the fact that $\mu'$ is the best matching that $w$ can obtain.
\end{proof}

\section{Omitted Proofs in Section~\ref{sec:many}}

\subsection{Proof of Theorem~\ref{thm:inc_many}}
\begin{reptheorem}{thm:inc_many}[Inconspicuous manipulation in the many-to-one setting]
    Given a many-to-one instance, for any $S_L$-Pareto-optimal matching $\mu$ computed by Algorithm \ref{algorithm} on the corresponding one-to-one instance, the same matching $\mu$ can be achieved on the many-to-one instance and each manipulator $i$ moves at most $q_i$ men to some higher rankings.
\end{reptheorem}

We will first prove the following lemma.

\begin{lemma}
	During the execution of Algorithm \ref{algorithm}, if a rotation $R$ contains a copy of a woman $w_i$, then it contains all $q_i$ copies of $w_i$.
\end{lemma}
\begin{proof}
	Consider the corresponding reduced lists of all men. Suppose the woman $w_{i,k}$ contained in $R$ is a copy of $w_i$. We assume, without loss of generality, that for all $1\le j\le q_i$, $w_{i,j}$ is matched to $m_j$ currently (i.e., $w_{i,j}$ is the first woman in $m_i$'s reduced list).
	
	For each $1\le j\le q_i-1$, we claim that $m_j\succ_{w_{i,j}}m_{j+1}$. This is because that $m_{j+1}$ has already been rejected by $w_{i,j}$ since he is now matched with $w_{i,j+1}$, and if $m_{j+1}\succ_{w_{i,j}}m_j$, $w_{i,j}$ cannot be matched with $m_j$ now since she once had a better partner $m_{j+1}$. Recall that all copies of $w_i$ has the same preference list. Thus $m_1\succ_{w_i}m_2\succ_{w_i}\dots\succ_{w_i}m_{q_i}$. This implies that $w_{i,j+1}$ through $w_{i,q_i}$ are still in the reduced list of $m_j$, and are ordered accordingly right after $w_{i,j}$. 
	
	Consider $m_k$ in rotation $R$. Since $w_{i,k}$ is the first woman in $m_k$'s reduced list, by definition, the next man in $R$ should have $w_{i,k}$ as the second woman in his reduced list, which is exactly $m_{k-i}$. Continuing with similar arguments, we know that $R$ contains a sequence $m_{q_i},\dots,m_2,m_1$, which indicates that all their matched women $w_{i,q_i}, \dots,w_{i,2},w_{i,1}$ are all contained in $R$.
\end{proof}

\begin{proof}[Proof of Theorem~\ref{thm:inc_many}.]
	To prove the theorem, we first focus on a specific way of constructing the $S_L$-Pareto-optimal strategy profile on the corresponding one-to-one instance. Then we show that based on this construction, we can construct an inconspicuous strategy profile on the original many-to-one instance (subject to the constraint that all copies of a woman have the same preference list) that gives the same matching.
	
	According to Lemma \ref{ReducedPropList}, to construct a $S_L$-Pareto-optimal strategy profile that yields $\mu$, we only need to construct $Pro_r(w_{i,j})$ for each manipulator $w_{i,j}$. Note that the proof of Lemma \ref{ReducedPropList} depends on the breadth-first search tree $T$ of the modified suitor graph, and $Pro_r(w_{i,j})$ contains exactly the two men who are the child ($\mu(w_{i,j}$) and the parent of the $w_{i,j}$ (it is easy to see that each woman only has one child and one parent in $T$).
	
	Now we construct a different tree $T'$ for the many-to-one setting such that Lemma \ref{ReducedPropList} still applies. Let $G$ be the modified suitor graph corresponding to the matching $\mu$. For any manipulator $w_{i,j}$ in $T$, there must be an edge $(m, w_{i,j})$ in $T$. $w_{i,j}$ cannot be matched to $m$ in $\mu$, since each woman has only 1 outgoing edge, and the previous edge in the path must come from $\mu(m)$. Let $D_k$ be the ordered two edges $(w_{i,k}, \mu(w_{i,k}))$ and $(\mu(w_{i,k}), w_{i,k-1})$. We show that we can replace edge $(m, w_{i,j})$ with a series of edges: $(m, w_{i, q_i}), D_{q_i}, D_{q_i-1},\dots, D_{j+1}$. 
	
	First, edge $(m, w_{i, q_i})$ is in graph $G$. To show this, note that edge $(m, w_{i,j})$ is in $G$ and $\mu(m)\ne w_{i,j}$, which indicates $w_{i,j}\succ_m \mu(m)$. So we have $w_{i,q_i} \succ_m \mu(m)$. According to the construction of the graph, we know that edge $(m, w_{i, q_i})$ is in $G$.
	
	Second, for each $D_k, j+1\le k\le q_i$, let $m'=\mu(w_{i,k})$. The first edge $(w_{i,k}, m')$ is in $G$ by definition. For the second edge $(m', w_{i,k-1})$, recall that we assume $w_{i,1}\succ_{m'} w_{i,2} \succ_{m'} \dots, \succ_{m'} w_{i,q_i}$, according to the definition of the modified suitor graph, $G$ contains the edge $(m', w_{i,k-1})$.
	
	Now we can construct the preference lists for all copies of manipulators with $T'$ according to Algorithm \ref{algorithm2}. In the resulting lists $P$, for copy $w_{i,j}, 1\le j\le q_i-1$ of a manipulator $w_i$, $w_{i,j}$ promotes man $\mu(w_{i,j+1})$ to the place right after man $\mu(w_{i,j})$. And $w_{i,q_i}$ promotes man $\mu(w_{i',1})$ to the place right after man $\mu(w_{i,q_i})$, where $w_{i'}$ is another woman. So we must have $\mu(w_{i,2})\succ_{w_i} \mu(w_{i,3})\succ_{w_i} \dots\succ_{w_i} \mu(w_{i,q_i})\succ_{w_i} \mu(w_{i',1})$ where $\succ_{w_i}$ is $w_i$'s true preference list. Now we re-construct a common preference list for all copies of $w_{i}$ by promoting $q_i$ men $\mu(w_{i,2}), \mu(w_{i,3}), \dots, \mu(w_{i,q_i}), \mu(w_{i',1})$ right after $\mu(w_{i,1})$. We claim that with this common preference list, the resulting matching is still $\mu$. To prove this, note that for $w_{i.j}$, the new common preference list can also be obtained by applying the following steps to the preference list $P(w_{i,j})$:
	\begin{enumerate}
		\item change the orders of men ranked above $\mu(w_{i,j})$;
		\item move the two men in $Pro_r(w_{i,j})$ to a higher rank;
		\item change the orders of men ranked below $\mu(w_{i,j})$ ($\mu(w_{i',1})$ if $j=q_i$).
	\end{enumerate}
	Clearly, the resulting matching does not change after the first step, since the men ranked above $\mu(w_{i,j})$ do not even propose to $w_{i,j}$. The other two steps also do not change the resulting matching since the proof of Theorem \ref{TinyChange} still holds.
\end{proof}

\section{Omitted Proofs in Section \ref{sec:strict_better}.}
\subsection{Proof of Theorem \ref{theo:hardness_better}.}
\begin{reptheorem}{theo:hardness_better}
	It is NP-complete to find a strategy profile, the induced matching of which is strictly better off for all manipulators.
\end{reptheorem} 

Clearly, this problem is in the NP class since given a preference profile, we can apply Gale-Shapley algorithm to generate the induced matching and verify the solution. In order to show the NP-completeness, we reduce {\sc 3-SAT} to this problem. Given an instance of {\sc 3-SAT} $\phi$, suppose the variable set is $V = \{x_1, \dots, x_n\}$, the corresponding literal set is $L = \{+x_i, -x_i ~|~ 1 \leq i \leq n\}$, and the clause set is $\{c_1, \dots, c_m\}$, where $c_j = (l_j^1, l_j^2, l_j^3)$. We construct an instance of our problem $G(\phi)$ with $N = 6 n + 2 m$ and
\begin{align*}
M = & ~ \{m_{x_i}^{+_1}, m_{x_i}^{+_2}, m_{x_i}^{+_3} ~|~ \forall 1 \leq i \leq n\} \cup \{m_{x_i}^{-_1}, m_{x_i}^{-_2}, m_{x_i}^{-_3} ~|~ \forall 1 \leq i \leq n\} \\ 
& \cup \{m_{c_j}^l  ~|~ \forall 1 \leq j \leq m\} \cup \{m_{c_j}^r  ~|~ \forall 1 \leq j \leq m\}, \\
W = & ~ \{w_{x_i}^{+_1}, w_{x_i}^{+_2}, w_{x_i}^{+_3} ~|~ \forall 1 \leq i \leq n\} \cup \{w_{x_i}^{-_1}, w_{x_i}^{-_2}, w_{x_i}^{-_3} ~|~ \forall 1 \leq i \leq n\} \\ 
& \cup \{w_{c_j}^l  ~|~ \forall 1 \leq j \leq m\} \cup \{w_{c_j}^r  ~|~ \forall 1 \leq j \leq m\}.
\end{align*}
The set of manipulators is
\[
L = \{w_{x_i}^{+_2} ~|~ \forall 1 \leq i \leq n\} \cup \{w_{x_i}^{-_2} ~|~ \forall 1 \leq i \leq n\} \cup \{w_{c_j}^r  ~|~ \forall 1 \leq j \leq m\}. \\ 
\]
The preference lists of each agent is specified as follows (the ``$\cdots$'' part at the end can be anything). For all $1 \leq i \leq n$ and each $x_i$, in the positive side (with superscript ``$+$''),
\begin{align*}
& P(m_{x_i}^{+_1}) = w_{x_i}^{+_1} \succ w_{x_i}^{+_2} \succ w_{x_i}^{-_3} \succ \cdots \\
& P(m_{x_i}^{+_2}) = w_{x_i}^{+_2} \succ w_{x_i}^{+_1} \succ \cdots \\
& P(w_{x_i}^{+_1}) = m_{x_i}^{+_2} \succ m_{x_i}^{+_1} \succ \cdots \\
& P(w_{x_i}^{+_2}) = m_{x_i}^{-_3} \succ m_{x_i}^{+_1} \succ m_{x_i}^{+_2} \succ m_{x_i}^{+_3} \succ \cdots \\
& P(w_{x_i}^{+_3}) = m_{x_i}^{-_1} \succ m_{x_i}^{+_3} \succ \cdots 
\end{align*}
In the negative side (with superscript $-$), similarly, 
\begin{align*}
& P(m_{x_i}^{-_1}) = w_{x_i}^{-_1} \succ w_{x_i}^{-_2} \succ w_{x_i}^{+_3} \succ \cdots \\
& P(m_{x_i}^{-_2}) = w_{x_i}^{-_2} \succ w_{x_i}^{-_1} \succ \cdots \\
& P(w_{x_i}^{-_1}) = m_{x_i}^{-_2} \succ m_{x_i}^{-_1} \succ \cdots \\
& P(w_{x_i}^{-_2}) = m_{x_i}^{+_3} \succ m_{x_i}^{-_1} \succ m_{x_i}^{-_2} \succ m_{x_i}^{-_3} \succ \cdots \\
& P(w_{x_i}^{-_3}) = m_{x_i}^{+_1} \succ m_{x_i}^{-_3} \succ \cdots 
\end{align*}
Suppose $+x_i \in c_{k_j}$ for all $1 \leq j \leq K_i^+$. The preference list of $m_{x_i}^{+_3}$ is
\[
P(m_{x_i}^{+_3}) = w_{x_i}^{+_2} \succ w_{x_i}^{+_3} \succ w_{c_{k_1}}^l \succ w_{c_{k_2}}^l \succ \cdots \succ w_{c_{k_{K_i^+}}}^l \succ w_{x_i}^{-_2} \succ \cdots
\]
Similarly, Suppose $-x_i \in c_{k_j}$ for all $1 \leq j \leq K_i^-$. The preference list of $m_{x_i}^{r_3}$ is
\[
P(m_{x_i}^{-_3}) = w_{x_i}^{-_2} \succ w_{x_i}^{-_3} \succ w_{c_{k_1}}^l \succ w_{c_{k_2}}^l \succ \cdots \succ w_{c_{k_{K_i^-}}}^l \succ w_{x_i}^{+_2} \succ \cdots
\]
Finally, we specify the preference lists for the agent with subscript $c_j$. For all $1 \leq j \leq m$,
\begin{align*}
& P(m_{c_j}^l) = w_{c_j}^l \succ w_{c_j}^r \succ \cdots \\
& P(m_{c_j}^r) = w_{c_j}^r \succ w_{c_j}^l \succ \cdots \\
& P(w_{c_j}^r) = m_{c_j}^l \succ m_{c_j}^r \succ \cdots
\end{align*}
Suppose $c_j = (s^1 ~ x_{j_1}) \lor (s^2 ~ x_{j_2}) \lor (s^3 ~ x_{i_3})$, where $s^1, s^2, s^3 \in \{-,+\}$. The preference list of $w_{c_j}^l$ is \footnote{If $s^k = +$, then $s^k_3 = +_3$; otherwise, if $s^k = -$, $s^k_3 = -_3$.}
\[
P(w_{c_j}^l) = m_{c_j}^r \succ m_{x_{j_1}}^{s_3^1} \succ m_{x_{j_2}}^{s_3^2} \succ m_{x_{j_3}}^{s_3^3} \succ m_{c_j}^l \succ \cdots
\]

To complete the reduction, we prove that $\phi$ is satisfiable if and only if $G(\phi)$ has a solution, i.e., there exists a strategy profile, whose induced matching is stable and strictly better off for all manipulators. 

First, notice the stable matching  $\mu$ generated by true preference lists is $\mu(m_{x_i}^{+_k}) = w_{x_i}^{+_k}$, $\mu(m_{x_i}^{-_k}) = w_{x_i}^{-_k}$ for all $1 \leq i \leq n$, $1 \leq k \leq 3$ and $\mu(m_{c_j}^l) = w_{c_j}^l$, $\mu(m_{c_j}^r) = w_{c_j}^r$ for all $1 \leq j \leq m$. Before providing proofs for both directions, we prove following lemmas first to establish intuitions.
\begin{lemma}\label{lem:better1}
	For all $i \in [n]$,  $w_{x_i}^{+_2}$ can perform a single-agent manipulation to be matched with $m_{x_i}^{+_1}$.
\end{lemma}
\begin{proof}
$w_{x_i}^{+_2}$ can manipulate her preference list to $P(w_{x_i}^{+_2}) = m_{x_i}^{-_3} \succ m_{x_i}^{+_1} \succ m_{x_i}^{+_3} \succ m_{x_i}^{+_2} \succ \cdots$.
\end{proof}

\begin{lemma}\label{lem:better2}
	For all $i \in [n]$,  $w_{x_i}^{+_2}$ can perform a single-agent manipulation to be matched with $m_{x_i}^{-_3}$.
\end{lemma}
\begin{proof}
$w_{x_i}^{+_2}$ can manipulate her preference list to $P(w_{x_i}^{+_2}) = m_{x_i}^{-_3} \succ m_{x_i}^{+_3} \succ m_{x_i}^{+_1} \succ m_{x_i}^{+_2} \succ \cdots$.
\end{proof}
By symmetry of construction, we have for each $1 \leq i \leq n$, woman $w_{x_i}^{-_2}$ can perform a single-agent manipulation to be matched with $m_{x_i}^{-_1}$ or $m_{x_i}^{+_3}$. 

\begin{lemma}\label{lem:combine_better}
	$w_{x_i}^{+_2}$ and $w_{x_i}^{-_2}$ cannot manipulate such that they are matched with $m_{x_i}^{-_3}$ and $m_{x_i}^{+_3}$ respectively at the same time, in any feasible permutation manipulation, while it is possible for them to manipulate to be matched with $m_{x_i}^{+_1}$ and $m_{x_i}^{-_1}$, $m_{x_i}^{-_3}$ and $m_{x_i}^{-_1}$, or, $m_{x_i}^{+_1}$ and $m_{x_i}^{+_3}$, respectively.
\end{lemma}

Before proving Lemma \ref{lem:combine_better}, we first prove the following lemma,
\begin{lemma}\label{lem:propose_woman}
	If the induced matching of a permutation manipulation on $G(\phi)$ is stable with respect to true preference lists, then 
	\begin{enumerate}
		\item For all $i \in [n]$, $m_{x_i}^{s_3}$, he cannot make proposals to any woman ranked below $w_{x_i}^{-s_2}$ in his true preference list; Moreover, he cannot be matched with any $w_{c_j}^l$;
		\item For all $i \in [n]$, $m_{x_i}^{s_1}$ and $m_{x_i}^{s_1}$ with $s \in \{+, -\}$, he can only make proposals to woman $w_{x_i}^{s'_k}$ with $s' \in \{+, -\}$ and $k \in \{1,2,3\}$;
		\item For all $j \in [m]$, both $m_{c_j}^l$ and $m_{c_j}^r$, he can only make proposals to $w_{c_j}^l$ and $w_{c_j}^r$.
	\end{enumerate}
\end{lemma}
\begin{proof}
Let
\[
W_i = \{w_{x_i}^{+_1}, w_{x_i}^{+_2}, w_{x_i}^{+_3}, w_{x_i}^{-_1}, w_{x_i}^{-_2}, w_{x_i}^{-_3}\}
\] and 
\[
M_i = \{m_{x_i}^{+_1}, m_{x_i}^{+_2}, m_{x_i}^{+_3}, m_{x_i}^{-_1}, m_{x_i}^{-_2}, m_{x_i}^{-_3}\}.
\]

First, for $m_{x_i}^{s_3}$ with $j \neq i$, $s \in \{+, -\}$, since $w_{x_i}^{-s_2}$ puts $m_{x_j}^{s_3}$ as the favorite candidate, if $m_{x_i}^{s_3}$ proposes to any woman ranked below $w_{x_i}^{-s_2}$ in his true preference list, the induced matching is unstable with respect to true preference lists. Moreover, if $m_{x_i}^{s_3}$ proposes to some $w_{c_j}^l$, then $w_{c_j}^l$ accepts $m_{x_{j_1}}^{s_3}$ and rejects $m_{c_j}^l$, next, $w_{c_j}^r$ accepts $m_{c_j}^l$ and rejects $m_{c_j}^r$, and finally, $w_{c_j}^l$ accepts $m_{c_j}^r$ and rejects $m_{x_{j_1}}^{s_3}$.

Second, except $m_{x_i}^{+_3}$ and $m_{x_i}^{-_3}$, all men in $M_i$ only propose to women in $W_i$ before they propose to the woman ranking him as the highest. Therefore, with similar arguments, we conclude that $m_{x_i}^{s_1}$ and $m_{x_i}^{s_1}$ with $s \in \{+, -\}$, he can only make proposals to woman $w_{x_i}^{s'_k}$ with $s' \in \{+, -\}$ and $k \in \{1,2,3\}$

Third, since $w_{c_j}^l$ ranks $m_{c_j}^r$ as favorite and $w_{c_j}^r$ ranks $m_{c_j}^l$ as favorite, according to the preference lists of $m_{c_j}^l$ and $m_{c_j}^r$, we can conclude they can only make proposals to $w_{c_j}^l$ and $w_{c_j}^r$;
\end{proof}

\begin{proof}[Proof of Lemma \ref{lem:combine_better}]
To achieve other combinations, $w_{x_i}^{+_2}$ and $w_{x_i}^{-_2}$ can manipulate their preference lists by following the manipulations in Lemma \ref{lem:better1} and Lemma \ref{lem:better2} according to their target partners. 

We prove the remaining case by contradiction. Suppose $w_{x_i}^{+_2}$ and $w_{x_i}^{-_2}$ can manipulate to a matching $\mu$ such that they are matched with $m_{x_i}^{-_3}$ and $m_{x_i}^{+_3}$. Then, since $w_{x_i}^{+_2}$ is matched with $m_{x_i}^{-_3}$, the closed set of rotations
\[
\left(\{m_{x_i}^{+_1}, m_{x_i}^{-_3}\}, \{w_{x_i}^{+_2}, w_{x_i}^{-_3}\}, \{w_{x_i}^{-_3}, w_{x_i}^{+_2}\}\right)
\] must be eliminated, which contains rotation 
\[
\left(\{m_{x_i}^{+_2}, m_{x_i}^{+_1}\}, \{w_{x_i}^{+_2}, w_{x_i}^{+_1}\}, \{w_{x_i}^{+_1}, w_{x_i}^{+_2}\}\right).
\] Similarly, since $w_{x_i}^{-_2}$ is matched with $m_{x_i}^{+_3}$, the closed set of rotations
\[
\left(\{m_{x_i}^{-_1}, m_{x_i}^{+_3}\}, \{w_{x_i}^{-_2}, w_{x_i}^{+_3}\}, \{w_{x_i}^{+_3}, w_{x_i}^{-_2}\}\right)
\] must be eliminated, which contains rotation 
\[
\left(\{m_{x_i}^{-_2}, m_{x_i}^{-_1}\}, \{w_{x_i}^{-_2}, w_{x_i}^{-_1}\}, \{w_{x_i}^{-_1}, w_{x_i}^{-_2}\}\right).
\] Therefore, all of $W_i = \{w_{x_i}^{+_1}, w_{x_i}^{+_2}, w_{x_i}^{+_3}, w_{x_i}^{-_1}, w_{x_i}^{-_2}, w_{x_i}^{-_3}\}$ have received more than one proposal. Moreover, according to Lemma \ref{lem:propose_woman}, they are matched with one of $M_i = \{m_{x_i}^{+_1}, m_{x_i}^{+_2}, m_{x_i}^{+_3}, m_{x_i}^{-_1}, m_{x_i}^{-_2}, m_{x_i}^{-_3}\}$. 

Henceforth, by Lemma \ref{CharOfProfile}, $\mu \in S_L$ only if there is some man $m \notin M_i$ having made a proposal to some $w \in W_i$ in order to create connections from $s$. However, according to Lemma \ref{lem:propose_woman}, if $\mu \in S_L$, no other man $m \notin M_i$ can make proposals to any $w \in W_i$.     
\end{proof}

According to this lemma, given an outcome of a manipulation, we construct the assignment as follows. $+x_i$ is assigned {\em true} if and only if $w_{x_i}^{+_2}$ is matched with $m_{x_i}^{-_3}$; otherwise, $-x_i$ is assigned {\em true}. Next lemma guarantees that such assignment is a satisfiable assignment for $\phi$.

\begin{lemma}\label{lem:betterif}
	For all $j \in [m]$, suppose $c_j = (s^1 ~ x_{j_1}) \lor (s^2 ~ x_{j_2}) \lor (s^3 ~ x_{j_3})$. Then, after manipulation, woman $w_{c_j}^r$ can be better off if and only if at least one $w_{x_{j_k}}^{s_2^k}$ is matched with $m_{x_{j_k}}^{-s_3^k}$ for $k \in \{1,2,3\}$.
\end{lemma}
\begin{proof}
{\bfseries The ``if'' direction}: Without loss of generality, suppose $w_{x_{j_1}}^{s_2}$ with $s = s^1$ is matched with $m_{x_{j_1}}^{-s_3}$ and thus, $m_{x_{j_1}}^{-s_3}$ has made proposal to $w_{x_{j_1}}^{-s_2}, w_{x_{j_1}}^{-s_3} \succ \cdots \succ w_{c_j}^l \succ \cdots \succ w_{x_{j_1}}^{s_2} \succ \cdots$. Thus, $w_{c_j}^l$ accepts $m_{x_{j_1}}^{-s_3}$ and rejects $m_{c_j}^l$, next, $w_{c_j}^r$ accepts $m_{c_j}^l$ and rejects $m_{c_j}^r$, and finally, $w_{c_j}^l$ accepts $m_{c_j}^r$ and rejects $m_{x_{j_1}}^{-s_3}$. Therefore, $w_{c_j}^r$ is better off. Moreover, if more than one $w_{x_{j_k}}^{s_2^k}$ is matched with $m_{x_{j_k}}^{-s_3^k}$, it does not change the matching of $w_{c_j}^r$ since she is already matched with her favorite one.

{\bfseries The ``only if'' direction}: If no $w_{x_{j_k}}^{s_2^k}$ is matched with $m_{x_{j_k}}^{-s_3^k}$, notice that no $m_{x_{j_k}}^{-s_3^k}$ makes proposal to $w_{c_j}^l$ since from argument in ``{\em if direction}'', we can see that if $m_{x_{j_k}}^{-s_3^k}$ makes proposal to $w_{c_j}^l$, $w_{x_{j_k}}^{s_2^k}$ is matched with $m_{x_{j_k}}^{-s_3^k}$. Therefore, if $w_{c_j}^r$ is better off, then $w_{c_j}^r$ is matched with $m_{c_j}^l$ and $w_{c_j}^l$ is matched with $m_{c_j}^r$, and notice that, $w_{c_j}^l, w_{c_j}^r$ have received more than one proposals. Henceforth, the matching after manipulation is in $S_L$ only if there is some man outside $m_{c_j}^l, m_{c_j}^r$ having made proposal to one of $w_{c_j}^l, w_{c_j}^r$ in order to create an edge pointing to the strongly connected component. However, according to Lemma \ref{lem:propose_woman}, we can conclude that no man outside $m_{c_j}^l, m_{c_j}^r$ having made proposal to one of $w_{c_j}^l, w_{c_j}^r$.
\end{proof}

With Lemma \ref{lem:better1}, \ref{lem:better2}, \ref{lem:combine_better} and \ref{lem:betterif}, we are ready to complete our reduction by showing $\phi$ is satisfiable if and only if $G(\phi)$ has a solution.

\begin{lemma}\label{lem:hard_onlyif}
	$\phi$ is satisfiable only if $G(\phi)$ has a solution.
\end{lemma}
\begin{proof}
Suppose $(l'_1, \dots, l'_n)$ is a satisfiable assignment. For all $i \in [n]$, 
\begin{enumerate}
	\item if $l'_i = +x_i$: $w_{x_i}^{+_2}$ manipulates to $m_{x_i}^{-_3}$ and $w_{x_i}^{-_2}$ manipulates to $m_{x_i}^{-_1}$;
	\item if $l'_i = -x_i$: $w_{x_i}^{+_2}$ manipulates to $m_{x_i}^{+_1}$ and $w_{x_i}^{-_2}$ manipulates to $m_{x_i}^{+_3}$.
\end{enumerate}
According to Lemma \ref{lem:combine_better}, the matching induced by this manipulation is in $S_L$. Moreover, since $(l'_1, \dots, l'_n)$ is a satisfiable assignment, from Lemma \ref{lem:betterif}, for all $j \in [m]$, $w_{c_j}^r$ is better off.
\end{proof}

\begin{lemma}\label{lem:hard_if}
	$\phi$ is satisfiable if $G(\phi)$ has a solution.    
\end{lemma}
\begin{proof}
From Lemma \ref{lem:combine_better}, for each $1 \leq i \leq n$, $w_{x_i}^{+_2}$ and $w_{x_i}^{-_2}$ cannot manipulate to be matched with $m_{x_i}^{-_3}$ and $m_{x_i}^{+_3}$ respectively. Therefore, we create the assignment as follows:
\begin{enumerate}
	\item $+x_i$ is assigned {\em true} if and only if $w_{x_i}^{+_2}$ is matched with $m_{x_i}^{-_3}$; 
	\item otherwise, $-x_i$ is assigned {\em true}.
\end{enumerate}
Moreover, from Lemma \ref{lem:betterif}, since for all $1 \leq j \leq m$ with $c_j = (s^1 ~ x_{j_1}) \lor (s^2 ~ x_{j_2}) \lor (s^3 ~ x_{j_3})$, $w_{c_j}^r$ is better off, at least one $w_{x_{j_k}}^{s_2^k}$ is matched with $m_{x_{j_k}}^{-s_3^k}$ for $k \in \{1,2,3\}$. Thus, the assignment we create must be a satisfiable assignment for $\phi$.
\end{proof}

\begin{lemma}\label{lem:NPhard_unStable}
	In our construction, if all manipulators are better off in a matching, the matching must be stable.
\end{lemma}
\begin{proof}
First, we point out that if a manipulation induces an unstable matching, then some woman must reject the best proposal she could have in the entire process. Henceforth, she must be a manipulator, while $L = \{w_{x_i}^{+_2} ~|~ \forall 1 \leq i \leq n\} \cup \{w_{x_i}^{-_2} ~|~ \forall 1 \leq i \leq n\} \cup \{w_{c_j}^r  ~|~ \forall 1 \leq j \leq m\}$ in our construction.

Notice that for all $1 \leq j \leq m$, $w_{c_j}^r$ can only be matched with $m_{c_j}^l$ if she is better off, and thus she cannot reject her best received proposal. 

The remaining manipulators are $w_{x_i}^{+_2}$ and $w_{x_i}^{-_2}$ for $1 \leq i \leq n$. Consider $w_{x_i}^{+_2}$ and the argument for $w_{x_i}^{-_2}$ is similar due to symmetry of construction. Since $w_{x_i}^{+_2}$ is better off, $w_{x_i}^{+_2}$ must be matched with either $m_{x_i}^{-_3}$ or $m_{x_i}^{+_1}$. In the case that she rejects her best received proposal, $w_{x_i}^{+_2}$ must be matched with $m_{x_i}^{+_1}$ and reject $m_{x_i}^{-_3}$. However, if $w_{x_i}^{+_2}$ is matched with $m_{x_i}^{+_1}$, $m_{x_i}^{+_1}$ stops proposing after meeting $w_{x_i}^{+_2}$, and thus, $w_{x_i}^{-_3}$ cannot reject $m_{x_i}^{-_3}$ since $w_{x_i}^{-_3}$ is a non-manipulator and she does not receive her favorite man $m_{x_i}^{+_1}$ to reject her second favorite man $m_{x_i}^{-_3}$. Therefore, $m_{x_i}^{-_3}$ has no chance to propose to $w_{x_i}^{+_2}$ and get rejected. 
\end{proof}

Theorem \ref{theo:hardness_better} follows from combining all the above results.  

\subsection{Proof of Theorem \ref{sharp_P}.}
\begin{reptheorem}{sharp_P}
	It is \#P-complete to compute the number of $S_L$-Pareto-optimal matchings, which are strictly better off for all manipulators.    
\end{reptheorem}

\begin{proof}
First of all, it is easy to check whether a matching is strictly better off and by using Algorithm~\ref{algorithm}, we can efficiently check whether a matching is $S_L$-Pareto-optimal. Therefore, this problem is in \#P.

Since computing the number of satisfiable assignment for {\sc 3-SAT} problem is \#P-complete, we only need to show that our reduction is {\em parsimonious}, i.e., the numbers of solutions in each problem are the same. 

Denoted by {\sc PARETO-BETTER} the problem of finding  $S_L$-Pareto-optimal matchings which are strictly better off for all manipulators. First, we show that given one satisfiable assignment for {\sc 3-SAT} problem, we can construct a solution to {\sc PARETO-BETTER}. According to Lemma \ref{lem:hard_onlyif}, we can construct a solution that makes all manipulators better off. Thus, it is sufficient to show that the constructed solution is also $S_L$-Pareto-optimal. In fact, for all $1 \leq i \leq n$, either $w_{x_i}^{+_2}$ or $w_{x_i}^{+_2}$ is matched with her favorite partner, but it is impossible for them to be matched with their favorite partners simultaneously. Moreover, for all $1 \leq j \leq m$, $w_{c_j}^r$ is matched with her favorite partner. Thus, such a solution must be $S_L$-Pareto-optimal.

Second, we show that given a solution to {\sc PARETO-BETTER}, we can construct a satisfiable assignment for {\sc 3-SAT} problem. From Lemma \ref{lem:hard_if}, we have shown that given a matching that makes all manipulators better off, we can construct a satisfiable assignment. Thus, given a solution to {\sc PARETO-BETTER}, we are also able to construct a satisfiable assignment for {\sc 3-SAT} problem.
\end{proof}

\end{document}